\documentclass[11pt]{article}

\usepackage[margin=1in]{geometry}
\usepackage{amsmath,amssymb,amsthm,mathtools}
\usepackage{bm}
\usepackage{booktabs}
\usepackage{graphicx}
\usepackage{xcolor}
\usepackage[section]{placeins}
\usepackage{enumitem}
\usepackage{microtype}
\usepackage[colorlinks=true,allcolors=blue]{hyperref}
\usepackage[nameinlink,noabbrev,capitalise]{cleveref}
\usepackage{orcidlink}
\usepackage{authblk}

\graphicspath{{./files/}}

\theoremstyle{plain}

\newtheorem{theorem}{Theorem}[section]
\newtheorem{proposition}{Proposition}[section]
\newtheorem{lemma}{Lemma}[section]
\newtheorem{corollary}{Corollary}[section]
\newtheorem{assumption}{Assumption}[section]

\newtheorem{remark}{Remark}[section]

\crefname{theorem}{Theorem}{Theorems}
\Crefname{theorem}{Theorem}{Theorems}

\crefname{proposition}{Proposition}{Propositions}
\Crefname{proposition}{Proposition}{Propositions}

\crefname{lemma}{Lemma}{Lemmas}
\Crefname{lemma}{Lemma}{Lemmas}

\crefname{corollary}{Corollary}{Corollaries}
\Crefname{corollary}{Corollary}{Corollaries}

\crefname{assumption}{Assumption}{Assumptions}
\Crefname{assumption}{Assumption}{Assumptions}

\crefname{example}{Example}{Examples}
\Crefname{example}{Example}{Examples}

\crefname{remark}{Remark}{Remarks}
\Crefname{remark}{Remark}{Remarks}

\newcommand{\R}{\mathbb{R}}
\newcommand{\Z}{\mathbb{Z}}
\newcommand{\E}{\mathbb{E}}
\newcommand{\Pp}{\mathbb{P}}
\newcommand{\F}{\mathcal{F}}

\newcommand{\cE}{\mathcal{E}}
\newcommand{\cI}{\mathcal{I}}
\newcommand{\bx}{\bm{x}}
\newcommand{\by}{\bm{y}}
\newcommand{\bz}{\bm{z}}
\newcommand{\ba}{\bm{a}}
\newcommand{\bp}{\bm{p}}
\newcommand{\bd}{\bm{d}}
\newcommand{\br}{\bm{r}}
\newcommand{\bq}{\bm{q}}
\newcommand{\bk}{\bm{k}}
\newcommand{\bi}{\bm{i}}
\newcommand{\bj}{\bm{j}}
\newcommand{\bs}{\bm{s}}
\newcommand{\bu}{\bm{u}}
\newcommand{\bA}{\bm{A}}
\newcommand{\bF}{\bm{F}}
\newcommand{\bK}{\bm{K}}
\newcommand{\bI}{\bm{I}}

\newcommand{\bD}{\bm{D}}
\newcommand{\bP}{\bm{P}}
\newcommand{\be}{\bm{e}}
\newcommand{\bn}{\bm{n}}
\newcommand{\dd}{\,\mathrm{d}}
\newcommand{\pv}{\operatorname{p.v.}}
\newcommand{\supp}{\operatorname{supp}}
\newcommand{\Qsub}[1]{Q_{#1}}
\newcommand{\veps}{\varepsilon}
\newcommand{\Qeps}{Q_\veps}
\newcommand{\Repsl}{R_{\veps,\ell}}

\newcommand\blfootnote[1]{%
  \begingroup
  \renewcommand\thefootnote{}\footnote{#1}%
  \addtocounter{footnote}{-1}%
  \endgroup
}

\title{Continuum Limit of Nonlocal Electrostatics in Random Media}
\author[1]{Prashant K. Jha\orcidlink{0000-0003-2158-364X}\thanks{Corresponding author.}}
\author[2]{Kaushik Dayal\orcidlink{0000-0002-0516-3066}}
\affil[1]{Department of Mechanical Engineering, South Dakota School of Mines and Technology,
Rapid City, SD 57701, USA.}
\affil[2]{Department of Civil and Environmental Engineering, Carnegie Mellon University,
Pittsburgh, PA 15213, USA.}
\date{}

\begin{document}
\maketitle

{
\small
\vspace{-20pt}
\begin{center}
    Published in {\it Mathematics \& Mechanics of Solids}\,; DOI: \url{https://doi.org/10.1177/10812865261494975}
\end{center}
\vspace{10pt}
}


\begin{abstract}
We derive a two-scale continuum limit for the electrostatic energy of random charge density fields that are stationary and ergodic under lattice translations. Each microscopic cell is charge neutral in every realization, while its dipole moment may fluctuate and have a nonzero mean. Under assumptions on the stationary microscopic potential and the finite-volume fields, the local and nonlocal energies converge almost surely to deterministic limits. The local limit consists of the ensemble-averaged microscopic Coulomb energy and a cell-depolarization term, whereas the effective polarization determines the nonlocal limit. The cell-depolarization terms cancel in the total energy, which consists of the ensemble-averaged microscopic cell energy and the macroscopic electrostatic field energy. The deterministic specialization recovers the periodic two-scale limit. For an independent-cell random-displacement model, we verify the assumptions of the continuum-limit theorem and show that centered microscopic fluctuations can change the local energy without changing the effective polarization or the nonlocal energy.
\end{abstract}
\blfootnote{Emails: \texttt{pjha.sci@gmail.com} (Prashant K. Jha) and \texttt{Kaushik.Dayal@cmu.edu} (Kaushik Dayal).}



\section{Introduction}
\label{sec:introduction}

Continuum electrostatic theories describe the energetic consequences of polarization while suppressing the microscopic charge distribution. Classical elastic-dielectric theories and their polarization-gradient extensions show how polarization, surfaces, and electromechanical coupling enter continuum energies \cite{Toupin1956,Mindlin1968,Tagantsev1986,ShenHu2010}. At the microscopic scale, however, Coulomb interactions are long ranged, and their lattice sums and boundary-dependent terms require special treatment \cite{deLeeuwPerramSmith1980,ToukmajiBoard1996}. The loss of microscopic information is especially relevant at finite temperature, where charged atoms fluctuate about their mean positions. Displaced-structure ensembles are used in finite-temperature electronic-structure calculations \cite{ZachariasGiustino2020}, and molecular simulations of ionic crystals resolve thermal motion about equilibrium positions \cite{GangemiEtAl2022}. Polarization fluctuations also contribute to dielectric response and interact with depolarizing fields and surface-charge screening \cite{PonomarevaBellaicheResta2007}. Finite-temperature fluctuations therefore motivate a continuum description that distinguishes their effect on local electrostatic interactions from their effect on the macroscopic polarization field.

We represent the microscopic fluctuations by a prescribed classical charge density $\rho=\rho(\bx,\by,\omega)$, where $\bx\in\Omega$, $\by\in\R^3$, and $\omega\in\Xi$ denote the material point, the microscopic coordinate, and a realization in a probability space $(\Xi,\F,\Pp)$, respectively. The long range of the Coulomb interaction makes the passage from the random microscopic field to a macroscopic energy nontrivial. Previous analyses of random Coulomb systems have treated thermodynamic limits for Coulomb matter and crystals \cite{LiebLebowitz1972,Fefferman1985}, energies and homogenization in stochastic lattices \cite{BlancLeBrisLions2007Energy,BlancLeBrisLions2007Homogenization}, and electronic-structure or mean-field models of disordered crystals \cite{BlancLewin2012,CancesLahbabiLewin2013}. Finite-temperature coarse graining has also been analyzed for one-dimensional models \cite{BlancEtAl2010}. Here, we obtain the continuum limit of the electrostatic energy for a prescribed stationary and ergodic random charge density in three dimensions.

Passing from a stationary microscopic field to a continuum energy requires spatial averaging. Stochastic homogenization provides such averaging for rapidly oscillating coefficients \cite{PapanicolaouVaradhan1981,Kozlov1980,DalMasoModica1986,JikovKozlovOleinik1994}, while periodic and stochastic two-scale convergence resolve both macroscopic and microscopic variables \cite{Nguetseng1989,Allaire1992,BourgeatMikelicWright1994}. Related homogenization and discrete-to-continuum results address stationary random measures, discrete networks, random convolution energies, and long-range multibody systems \cite{ZhikovPiatnitskii2006,NeukammVarga2018,BraidesPiatnitski2021,BachBraidesCicalese2020}. We use lattice ergodic averaging within a two-scale energy decomposition. Unlike standard local homogenized energies, the Coulomb scaling retains a nonlocal contribution determined by the macroscopic polarization.

Marshall and Dayal \cite{MarshallDayal2014} introduced the local/nonlocal energy separation for periodic charge densities, distinguishing interactions within each material point from polarization-mediated interactions between distinct material points. Their formulation was subsequently used in atomistic-to-continuum models of crystalline defects \cite{jha2022atomic}. This decomposition of total energy was used to perform formal calculations for random charge fields using an $\R^3$-parameterized action in \cite{Jha2016}. Long-range electrical interactions in nanostructures \cite{JhaBreitzmanDayal2023} and dipolar limits in micromagnetics \cite{JamesMuller1994,MullerSchlomerkemper2002,SchlomerkemperSchmidt2009} provide related examples in which lattice-scale regularization produces both macroscopic fields and local corrections. Recent two-scale analyses of periodic dielectric crystals address the dependence of bulk polarization on the reference-cell convention and the accompanying boundary charge \cite{SenWangBreitzmanDayal2024,SenWangBreitzmanDayal2026}.

Using the two-scale strategy and local/nonlocal decomposition of \cite{MarshallDayal2014}, we derive an effective continuum-limit expression for the electrostatic energy of stationary and ergodic random charge density fields and prove almost-sure convergence to a deterministic limit under the stated assumptions. The three length scales are illustrated in \cref{fig:length_scales}: microscopic charge distributions at scale $\ell$ are grouped into material points of size $\veps$ within a body of characteristic size $L$, with $\ell\ll\veps\ll L$.

\begin{figure}[!t]
\centering
\includegraphics[width=0.95\textwidth]{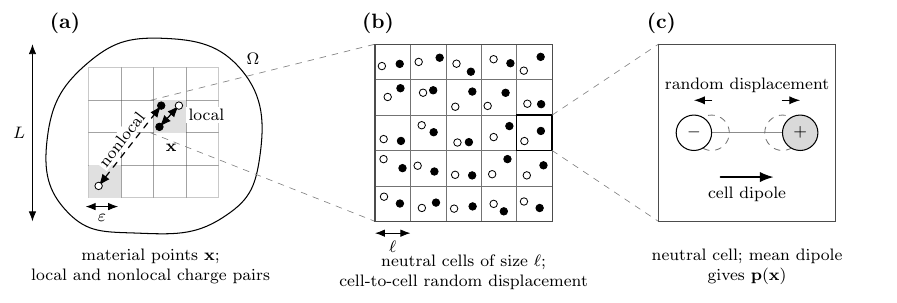}
\vspace{-10pt}
\caption{Length-scale hierarchy underlying the local/nonlocal energy decomposition: (a) local and nonlocal interactions, (b) neutral cells, and (c) cell dipoles.}
\label{fig:length_scales}
\end{figure}

In the random setting, stationarity under lattice translations is expressed through measurable, measure-preserving, ergodic actions $\{T_{\bk}:\Xi\to\Xi\}_{\bk\in\Z^3}$. A reference-cell charge density $\bar\rho$ generates the stationary microscopic field according to
\begin{equation}
    \rho(\bx,\by,\omega)
    =
    \bar\rho(\bx,\bs,T_{\bk}\omega),
    \qquad
    \by=\bk+\bs,
    \label{eq:stationary_representation_intro}
\end{equation}
where $\bk\in\Z^3$ identifies a cell and $\bs\in Q:=[-1/2,1/2)^3$ is the coordinate within the reference cell. Exact charge neutrality is imposed on $\bar\rho$ in every realization and is inherited by each translated cell. Within a realization, different cells may have different charge profiles and dipole moments, while every cell remains charge neutral.

The lattice-stationary representation in \cref{eq:stationary_representation_intro} requires a reference-cell convention, and we take $Q$ to be the unit cube. The cell polarization depends on how we assign charges to $Q$. The resulting nonuniqueness of bulk polarization and its relation to surface charge are central to the modern theory of polarization \cite{KingSmithVanderbilt1993,VanderbiltKingSmith1993,Resta1994,RestaVanderbilt2007,Spaldin2012}. In the periodic setting, a change in the reference-cell convention changes the bulk polarization together with the corresponding boundary charge, leaving the physical electric field and total electrostatic energy unchanged \cite{SenWangBreitzmanDayal2024,SenWangBreitzmanDayal2026}. We therefore state the cell polarization and the associated local quantities relative to the fixed convention $Q$.

Cell neutrality removes the contribution of the lattice position to the dipole of each translated cell. The lattice ergodic theorem then identifies the large-volume average of the cell dipoles as the deterministic effective polarization
\begin{equation*}
    \bp(\bx)
    =
    \E\!\left[
    \int_Q
    \bar\rho(\bx,\bs,\cdot)\bs\dd\bs
    \right].
\end{equation*}
Individual cell dipoles may vary within a realization, but their large-volume average is independent of the realization.

Combining the finite-scale energy decomposition with the ergodic identification of $\bp$ gives the main result. Under the assumptions in \cref{sec:assumptions}, the local contribution converges almost surely to
\begin{equation*}
    \cE_{\mathrm{loc}}
    =
    \frac12
    \int_\Omega
    \E\!\left[
    \int_Q
    \bar\rho(\bx,\bs,\cdot)
    \bar h(\bx,\bs,\cdot)
    \dd\bs
    \right]\dd\bx
    +
    \frac16
    \int_\Omega
    |\bp(\bx)|^2
    \dd\bx,
\end{equation*}
where $\bar h$ is the reference-cell representation of the lattice-stationary microscopic Coulomb potential; see \cref{eq:reference_cell_potential}. The first term retains the statistics of the microscopic charge density, while the second is the depolarization energy associated with the cubic material point. The interactions between distinct material points converge to the nonlocal dipole--dipole energy, equivalently represented by
\begin{equation*}
    \cE_{\mathrm{nloc}}
    =
    \frac12
    \int_{\R^3}
    |\nabla\phi|^2
    \dd\bx
    -
    \frac16
    \int_\Omega
    |\bp(\bx)|^2
    \dd\bx,
\end{equation*}
where $\phi\in\dot H^1(\R^3)$ solves $-\Delta\phi=-\nabla\cdot\bp_0$ in $\R^3$ and $\bp_0$ is the extension of $\bp$ by zero outside $\Omega$. The equivalent principal-value representation using the dipole kernel is stated in \cref{prop:dipole_representation}. The $+\frac16\int_\Omega|\bp|^2$ local term and the corresponding negative nonlocal term cancel. Let $\cE_{\veps,\ell}(\omega)$ denote the discrete electrostatic energy for the realization $\omega$. The cancellation yields, for almost every realization,
\begin{equation*}
    \lim_{\veps\downarrow0}
    \lim_{\ell/\veps\downarrow0}
    \cE_{\veps,\ell}(\omega)
    =
    \frac12
    \int_\Omega
    \E\!\left[
    \int_Q
    \bar\rho(\bx,\bs,\cdot)
    \bar h(\bx,\bs,\cdot)
    \dd\bs
    \right]\dd\bx +
    \frac12
    \int_{\R^3}
    |\nabla\phi|^2
    \dd\bx.
\end{equation*}
The local and nonlocal limits retain different information about microscopic fluctuations. Although both limits are deterministic, the local limit retains the ensemble-averaged microscopic charge statistics, whereas the nonlocal limit is determined solely by the effective polarization.

The effective expression in \cref{eq:main_energy_limit} allows the limiting energy to be evaluated without resolving and summing Coulomb interactions throughout an increasingly large random specimen. Direct and accelerated Coulomb calculations operate on microscopic charge configurations \cite{deLeeuwPerramSmith1980,ToukmajiBoard1996,GreengardRokhlin1987}. Here, once the microscopic energy expectation and effective polarization have been determined for the prescribed random field, the remaining calculation consists of macroscopic integration and solution of the Poisson problem for $\phi$.

To show that the assumptions admit both periodic and nonperiodic random charge fields, we examine two cases. When the probability space consists of a single point, the framework reduces to a periodic charge density, and the continuum energy agrees with the formal periodic energy obtained in \cite{MarshallDayal2014}. We then verify the assumptions for a concrete nonperiodic random-displacement model in which the displacement fields are identically distributed and independent from cell to cell, while displacements of charges within a cell may be correlated and may vary with the macroscopic position $\bx$. \Cref{sec:numerical_illustration} compares finite-scale energies and polarizations with the limiting expressions for a two-charge random-displacement model.

For centered displacements, microscopic fluctuations can change the ensemble-averaged microscopic term in the local energy while leaving the effective polarization, and hence the nonlocal energy, unchanged. The independent-cell measure provides an idealized static ensemble rather than an atomistic Gibbs measure or a model of temporal dynamics; physical displacements are generally correlated across cells.

The remainder of the paper is organized as follows. \Cref{sec:microscopic_model} defines the two-scale charge density, the finite-scale electrostatic energy, and the lattice-stationary averaging framework. \Cref{sec:assumptions} states the assumptions on the charge density and microscopic fields. The continuum limit is established in \cref{sec:continuum_limit}. The periodic and random-displacement cases are considered in \cref{sec:examples}, and \cref{sec:conclusions} summarizes the main conclusions and scope of the result.

\section{Two-scale electrostatic model}
\label{sec:microscopic_model}

We now define the finite-scale electrostatic model used in the analysis. We construct the random charge density field on the two-scale domain using an ergodic lattice action and, following \cite{MarshallDayal2014}, separate the total electrostatic energy into local and nonlocal contributions. We also introduce the charge moments and interaction potentials used in the subsequent analysis. The section concludes with the ergodic theorems for lattice-stationary random fields that are needed to pass to the continuum limit.

\subsection{Material points and microscopic coordinates}\label{sec:setup}

Let $\Omega\subset\R^3$ be a bounded Lipschitz domain whose characteristic length is normalized to $L=1$. As illustrated in \cref{fig:length_scales}, we first represent the material at the scale $\veps$, over which material properties such as the charge density field vary slowly. The behavior of each material point is determined by the microscopic interactions within its associated region at the smaller length scale $\ell$. Thus,
\begin{equation*}
\ell\ll\veps\ll1.
\end{equation*}

To represent the discrete structure of the material at the scale $\veps$, we introduce the set of material points
\begin{equation}
\Omega_\veps
:=
\left\{
\bx\in\veps\Z^3\cap\Omega:
\bx+\Qeps\subset\overline\Omega
\right\},
\qquad
\Qeps:=\veps Q,
\qquad
Q:=\left[-\frac12,\frac12\right)^3.
\label{eq:cells}
\end{equation}
Here, $\Omega_\veps$ is the set of material points in $\Omega$, and $\bx+\Qeps$ is the corresponding material-point region. We include a region when it lies in $\overline\Omega$; touching $\partial\Omega$ does not change the energy integrals. The regions are disjoint, their union lies in $\overline\Omega$, and the part of $\Omega$ they do not cover is contained in a boundary layer whose volume vanishes as $\veps\to0$.

Within each material-point region, we consider a lattice of size $\ell$ containing charged atoms. Let $\bz_i$ denote the physical position of the $i$th charged atom, let $q_i$ denote its charge, and let $\cI_{\bx}$ denote the set of atoms contained in $\bx+\Qeps$. The electrostatic energy of the material can then be written as
\begin{equation*}
\hat{\cE}_{\veps,\ell}
=
\frac12
\sum_{\bx,\bx'\in\Omega_\veps}
\sum_{\substack{
i\in\cI_{\bx},
j\in\cI_{\bx'} \\
(\bx,i) \ne (\bx',j)}}
G(\bz_i-\bz_j)q_iq_j,
\qquad
G(\bz):=\frac{1}{4\pi|\bz|}.
\end{equation*}
Here, $G(\cdot)$ is the Coulomb kernel. We suppress the vacuum permittivity, which can be restored by multiplying the energy by the appropriate material coefficient. The factor of $1/2$ accounts for double counting of the interactions. The outer sum is over the material points, while the inner sum accounts for the interactions between the discrete charges contained in the corresponding material-point regions. The condition $(\bx,i)\ne(\bx',j)$ excludes the self-interaction of an individual charge.

In this work, we use an alternative description of the electrostatic energy in terms of a charge density field $\hat{\rho}(\bx,\br)$; see \cite{MarshallDayal2014,jha2022atomic,xiao2005influence}. Here, $\br\in \Qeps$ denotes the position relative to the material point $\bx$, so that $\bx+\br$ is the corresponding physical position. The field $\hat{\rho}(\bx,\br)$ provides a continuous representation of the discrete charges contained in $\bx+\Qeps$. In terms of this field, the electrostatic energy is
\begin{equation*}
\hat{\cE}_{\veps,\ell}
=
\frac12
\sum_{\bx,\bx'\in\Omega_\veps}
\int_{Q_\veps}
\int_{Q_\veps}
G(\bx+\br-\bx'-\br')
\hat{\rho}(\bx,\br)
\hat{\rho}(\bx',\br')
\dd\br'\dd\br.
\end{equation*}
This form is more convenient for the analysis of the continuum limit. The discrete lattice description is recovered by taking
\begin{equation*}
\hat{\rho}(\bx,\br)
=
\sum_{i\in{\cI}_{\bx}}
q_i 
\delta\left(\br-(\bz_i-\bx)\right),
\end{equation*}
where $\delta$ is the Dirac delta distribution. When using this charge density representation for point charges, we must ensure the summation excludes $\bz_i - \bz_j = 0$ to avoid the singularity of the Coulomb kernel. 

Next, let $\by$ denote the microscopic point, which is obtained by rescaling the coordinate $\br$ by the lattice scale $\ell$:
\begin{equation*}
    \by=\frac{\br}{\ell}
    \in \Qsub{\Repsl},
    \qquad
    \Repsl:=\frac{\veps}{\ell}.
\end{equation*}
Since $\Repsl>1$, for a given $\by \in \Qsub{\Repsl}$, there are unique $\bk \in \Z^3$ and $\bs \in Q$ such that $\by = \bk + \bs$, 
where $\bk$ identifies a microscopic cell and $\bs$ is the coordinate within the reference cell.

The thermodynamic limit within a material point corresponds to $\Repsl\to\infty$ at fixed $\veps$. For convenience, we take this limit along the sequence
\begin{equation}
    \ell_n=\frac{\veps}{2n+1}, \qquad 
R_{\veps, \ell_n} = \frac{\veps}{\ell_n} = 2n+1 =: R_{n}\,.
\label{eq:lattice_scale_sequence}
\end{equation}
The microscopic region $\Qsub{R_n}$ is the union of $(2n+1)^3$ reference cells $Q$ arranged symmetrically about the origin. The index set of the lattice cells contained in $\Qsub{R_n}$ is
\begin{equation}
I_{R_n}:=\{-n,\ldots, 0, \ldots, n\}^3.
\label{eq:lattice_index_set}
\end{equation}
Consequently,
\begin{equation}
Q_{R_n}
=
\biguplus_{\bk\in I_{R_n}}(\bk+Q),
\qquad
|Q_{R_n}|=|I_{R_n}|=R_n^3.
\label{eq:cell_union}
\end{equation}
Thus, $\by \in Q_{R_n}$ is the scaled relative coordinate of the microscopic point $\br = \ell_n \by \in \Qeps$.

\subsection{Random charge density field}\label{sec:randomDensity}

Our goal is to describe a material in which the microscopic charge density profile can vary from one cell to another in a given realization, while its statistical properties remain unchanged under lattice translations. We also require averages over a large number of microscopic cells to converge to the corresponding ensemble averages. These properties are described using measurable, measure-preserving, ergodic lattice actions (dynamical systems), a standard framework in the analysis of random media and stochastic homogenization \cite{Kozlov1980,JikovKozlovOleinik1994}.

Let $(\Xi,\F,\Pp)$ be a probability space, where $\omega\in\Xi$ denotes a
realization of the microscopic charge density. Let
$\{T_{\bk}:\Xi\to\Xi\}_{\bk\in\Z^3}$ be a measurable, measure-preserving
lattice action satisfying
\begin{equation}
    T_{\bm{0}}=\operatorname{Id},
    \qquad
    T_{\bk+\bk'}=T_{\bk}\circ T_{\bk'},
    \qquad
    \Pp(T_{\bk}^{-1}A)=\Pp(A)
    \label{eq:action_properties}
\end{equation}
for every $\bk,\bk'\in\Z^3$ and $A\in\F$. The action is assumed to be ergodic. That is, if a measurable set
$A\in\F$ satisfies
\begin{equation*}
    T_{\bk}^{-1}A=A
    \qquad
    \text{for every }\bk\in\Z^3,
\end{equation*}
then $\Pp(A)$ is either zero or one. Ergodicity allows spatial averages over
a large number of cells to be identified with ensemble averages \cite[Chapter 7]{JikovKozlovOleinik1994}.

Using the lattice action, we define the random charge density profile on the reference cell $Q$ by
\begin{equation*}
    \bar\rho:\overline\Omega\times Q\times\Xi\to\R.
\end{equation*}
For fixed $\bx\in\overline\Omega$ and $\omega\in\Xi$, the function
$\bar\rho(\bx,\cdot,\omega)$ describes the charge density within the
reference cell $Q$. The charge density on the full lattice is obtained by extending $\bar\rho$ using the maps $\{T_{\bk}\}$ as follows:
\begin{equation}\label{eq:stationary_representation}
\rho(\bx, \by, \omega) := \bar\rho(\bx, \bs, T_{\bk}\omega)\,,\qquad \forall \by \in \R^3\,,
\end{equation}
where $\bk\in \Z^3$ and $\bs \in Q$ are such that $\by = \bk + \bs$. Thus, for a fixed realization $\omega$, different cells can have different charge density profiles, depending on the map $T_{\bk}$ associated with the cell coordinate $\bk$. However, because the action is measure preserving, the charge density profiles have the same probability distribution in every cell.

Within the material-point region $\bx + Q_\veps$, the physical charge density $\rho^\ell$ is defined by
\begin{equation}
    \rho^\ell(\bx,\br,\omega)
    :=
    \frac{1}{\ell}
    \rho\left(\bx,\frac{\br}{\ell},\omega\right),
    \qquad
    \br\in Q_\veps.
    \label{eq:scaled_charge_density}
\end{equation}
The ansatz scaling factor $\ell^{-1}$ is introduced so that the local and nonlocal electrostatic energies are of order one in the limit $\ell/\veps \to 0$. For a fixed reference profile, the integral of $|\rho^\ell|$ over a microscopic cell scales as $\ell^2$, while its net charge is zero. Here, $\bar\rho$ generates the lattice-stationary microscopic field $\rho$ and, after rescaling, the physical charge density $\rho^\ell$ used in the finite-scale electrostatic energy.

\begin{remark}
The choice of the generator $\bar{\rho}: \overline\Omega \times Q \times \Xi \to \R$ and the ergodic lattice action $\{T_{\bk}\}_{\bk\in\Z^3}$ allows a broad class of random charge density fields. Periodic charge density fields and random perturbations of periodic fields are special cases; see \cref{sec:examples}.
\end{remark}

\subsection{Discrete electrostatic energies}

Using the physical charge density $\rho^\ell$ defined in
\cref{eq:scaled_charge_density}, the electrostatic energy at fixed
$\veps$ and $\ell$ is
\begin{equation}\label{eq:finite_total_energy}
    \cE_{\veps,\ell}(\omega)
    :=
    \frac12
    \sum_{\bx,\bx'\in\Omega_\veps}
    \int_{Q_\veps}
    \int_{Q_\veps}
    \rho^\ell(\bx,\br,\omega)
    \rho^\ell(\bx',\br',\omega)
    \,
    G(\bx+\br-\bx'-\br')
    \dd\br'\dd\br.
\end{equation}
We separate the energy according to whether the two interacting charge
distributions belong to the same material point or to distinct material points. The local energy contains the interactions within each material-point region:
\begin{equation}\label{eq:finite_local_energy}
    \cE_{\veps,\ell}^{\mathrm{loc}}(\omega)
    :=
    \frac12
    \sum_{\bx\in\Omega_\veps}
    \int_{\veps Q}
    \int_{\veps Q}
    \rho^\ell(\bx,\br,\omega)
    \rho^\ell(\bx,\br',\omega)
    \,
    G(\br-\br')
    \dd\br'\dd\br.
\end{equation}
The nonlocal energy contains the interactions between distinct
material-point regions:
\begin{equation}\label{eq:finite_nonlocal_energy}
    \cE_{\veps,\ell}^{\mathrm{nloc}}(\omega)
    :=
    \frac12
    \sum_{\substack{
    \bx,\bx'\in\Omega_\veps\\
    \bx\ne\bx'}}
    \int_{\veps Q}
    \int_{\veps Q}
    \rho^\ell(\bx,\br,\omega)
    \rho^\ell(\bx',\br',\omega)
    \,
    G(\bx+\br-\bx'-\br')
    \dd\br'\dd\br.
\end{equation}
Therefore,
\begin{equation}
    \cE_{\veps,\ell}
    =
    \cE_{\veps,\ell}^{\mathrm{loc}}
    +
    \cE_{\veps,\ell}^{\mathrm{nloc}}.
    \label{eq:energy_split}
\end{equation}

\subsection{Microscopic moments and interaction potentials}

We introduce quantities that will appear in the limiting continuum model. For each material point $\bx\in\Omega_\veps$, define the total charge and dipole moment by
\begin{equation}\label{eq:cell_moments}
    q_{\veps,\ell}(\bx,\omega)
    :=
    \int_{\veps Q}
    \rho^\ell(\bx,\br,\omega)\dd\br, \qquad 
    \bd_{\veps,\ell}(\bx,\omega)
    :=
    \int_{\veps Q}
    \rho^\ell(\bx,\br,\omega)\br\dd\br.
\end{equation}
Using $\by=\br/\ell$ and
$R=R_{\veps,\ell}=\veps/\ell$, the charge and dipole densities can be written using \eqref{eq:scaled_charge_density} as
\begin{equation}\label{eq:rescaled_cell_moments}
    \frac{q_{\veps,\ell}(\bx,\omega)}{\veps^3}
    =
    \frac{R}{\veps}
    \frac{1}{|Q_R|}
    \int_{Q_R}
    \rho(\bx,\by,\omega)\dd\by,
    \qquad 
    \bp_R(\bx,\omega)
    := \frac{\bd_{\veps,\ell}(\bx,\omega)}{\veps^3}
    =
    \frac{1}{|Q_R|}
    \int_{Q_R}
    \rho(\bx,\by,\omega)\by\dd\by.
\end{equation}
\cref{asmp:neutrality} implies that $q_{\veps,\ell}(\bx,\omega)=0$
for every fixed $\bx$ and almost every $\omega$. Under the charge neutrality condition, the following will be shown
\begin{equation}
    \bp_R(\bx,\omega) \quad 
    \xrightarrow[R\to\infty]{} \quad 
    \bp(\bx)
    :=
    \E\left[
    \int_Q
    \rho(\bx,\by,\cdot)\by\dd\by
    \right]\,,
    \label{eq:effective_polarization}
\end{equation}
where $\bp(\bx)$ is the effective polarization. 

We next introduce the microscopic potential used in the local-energy
analysis. For each $\bx\in\overline\Omega$ and $\omega\in\Xi$, define
\begin{equation}
    h(\bx,\by,\omega)
    :=
    G\ast\rho(\bx,\cdot,\omega)(\by),
    \label{eq:microscopic_potential}
\end{equation}
whenever the infinite-volume convolution exists. \cref{asmp:stationary} specifies its existence and the integrability required in the analysis.

The microscopic potential $h$ is a lattice-stationary random field, which follows from the lattice-stationary property of the charge density field $\rho$. To see this, we define the restriction of $h$ to the reference cell by
\begin{equation}
    \bar h(\bx,\bs,\omega)
    :=
    h(\bx,\bs,\omega),
    \qquad
    \bs\in Q.
    \label{eq:reference_cell_potential}
\end{equation}
It follows that, for $\by \in Q$ and $\bk \in \Z^3$,
\begin{equation*}
\begin{aligned}
    h(\bx,\by+\bk,\omega)
    &=
    G\ast\rho(\bx,\cdot,\omega)(\by+\bk)
    =
    G\ast\rho(\bx,\cdot+\bk,\omega)(\by)\\
    &=
    G\ast\rho(\bx,\cdot,T_{\bk}\omega)(\by)
    =
    h(\bx,\by,T_{\bk}\omega) = \bar{h}(\bx,\by,T_{\bk}\omega).
\end{aligned}
\end{equation*}

At finite volume ($\ell/\veps > 0$), the charge density within a material point is restricted to the finite microscopic region $Q_R$ ($R = \Repsl = \veps/\ell$). The corresponding truncated potential is
\begin{equation}
    h_R(\bx,\by,\omega)
    := G\ast (\rho \chi_{Q_R}) = 
    \int_{Q_R}
    G(\by-\bz)\rho(\bx,\bz,\omega)\dd\bz.
    \label{eq:truncated_potential_at_origin}
\end{equation}
Unlike $h$, the potential $h_R$ is not lattice stationary because its source is restricted to the fixed cube $Q_R$. The difference $h_R-h$ need not vanish in the limit and accounts for the correction arising from the boundary of the material-point region. The limit is described by the potential of a uniformly polarized cube. For $\bq\in\R^3$, define
\begin{equation}
    H_Q[\bq](\boldsymbol\xi)
    :=
    \int_{\partial Q}
    G(\boldsymbol\xi-\boldsymbol\eta)
    \bq\cdot\bn_Q(\boldsymbol\eta)
    \dd S_{\boldsymbol\eta}.
\label{eq:cube_depolarization_potential}
\end{equation}
It satisfies
\begin{equation}
-\Delta H_Q[\bq]
    =
    -\nabla\cdot(\bq\chi_Q)
    \qquad
    \text{in }\mathcal D'(\R^3).
\end{equation}
The source $-\nabla\cdot(\bq\chi_Q)$ is the charge source due to the uniform polarization $\bq$.

For a fixed $\bk\in\Z^3\setminus\{\bm0\}$, we define the interaction
between two microscopic regions separated by $R\bk$ as
\begin{equation}
    B_{R,\bk}(\bx,\bx',\omega)
    :=
    \frac{1}{R^3}
    \int_{Q_R}
    \rho(\bx',\bz,\omega)
    h_R(\bx,\bz+R\bk,\omega)
    \dd\bz.
    \label{eq:finite_volume_pair_interaction}
\end{equation}
Using \cref{eq:truncated_potential_at_origin} and the symmetry
$G(-\bz)=G(\bz)$, this is equivalently
\begin{equation*}
\begin{aligned}
    B_{R,\bk}(\bx,\bx',\omega)
    =
    \frac{1}{R^3}
    \int_{Q_R}\int_{Q_R}
    G(\by-\bz-R\bk)
    \rho(\bx,\by,\omega)
    \rho(\bx',\bz,\omega)
    \dd\by\dd\bz.
\end{aligned}
\end{equation*}
The corresponding interaction between uniformly polarized unit cubes is
defined by
\begin{equation}
\begin{aligned}
    H_{\bk+Q}[\bq](\boldsymbol\xi)
    &:=
    H_Q[\bq](\boldsymbol\xi-\bk),\\
    B_{\bk}(\bp,\bq)
    &:=
    \int_{\R^3}
    \nabla H_Q[\bp]\cdot
    \nabla H_{\bk+Q}[\bq]
    \dd\boldsymbol\xi.
\end{aligned}
\label{eq:polarized_cell_pair_interaction}
\end{equation}

At large separations, the interaction between the polarized cubes reduces
to the interaction between their dipole moments. The dipole interaction
kernel is
\begin{equation}
    \bK(\bz)
    :=
    -\nabla^2G(\bz)
    =
    \frac{1}{4\pi|\bz|^3}
    \left(
    \bI
    -
    3\frac{\bz}{|\bz|}
    \otimes
    \frac{\bz}{|\bz|}
    \right),
    \qquad
    \bz\ne\bm0.
    \label{eq:dipole_kernel}
\end{equation}

\subsection{Ergodic theorem for lattice-stationary random variables}
\label{sec:random_fields}

The action $T_{\bk}$ used in this work is parameterized by $\Z^3$. For an $\R^3$-parameterized action, the Birkhoff ergodic theorem identifies the spatial average of a stationary random field with its expectation \cite[Theorem~7.2]{JikovKozlovOleinik1994}. The corresponding spatial average for the lattice action is the discrete average over the cells indexed by $I_R$. We use the following result, which is a special case of \cite[Theorem~7.2]{JikovKozlovOleinik1994}. We also state a corollary that is used in the analysis of the continuum limit.

\begin{theorem}[Lattice ergodic theorem]
\label{thm:ergodic}
Let $f\in L^1(\Xi; \R^m)$, and let
$\{T_{\bk}\}_{\bk\in\Z^3}$ be the measure-preserving ergodic action in
\cref{eq:action_properties}. Then, for almost every $\omega\in\Xi$,
\begin{equation}
    \lim_{\substack{R\to\infty\\R=2n+1}}
    \frac{1}{R^3}
    \sum_{\bk\in I_R}
    f(T_{\bk}\omega)
    =
    \E[f].
    \label{eq:ergodic_average}
\end{equation}
If $f\in L^2(\Xi; \R^m)$, the convergence holds in the $L^2$ norm.
\end{theorem}

\begin{corollary}[Weighted lattice ergodic limit]
\label{cor:weighted_ergodic}
Let $a\in C(\overline Q)$ and let
$f\in L^2(\Xi;\R^m)$. Then, for almost every $\omega\in\Xi$,
\begin{equation}
    \frac{1}{R^3}
    \sum_{\bk\in I_R}
    a\left(\frac{\bk}{R}\right)
    f(T_{\bk}\omega)
    \longrightarrow
    \left(\int_Q a(\boldsymbol\eta)\dd\boldsymbol\eta\right)
    \E[f]
\end{equation}
as $R=2n+1\to\infty$.
\end{corollary}

\cref{thm:ergodic} and \cref{cor:weighted_ergodic} are proved in \cref{sec:ergodicityProof}. The result applies componentwise to finite-dimensional vector-valued random variables and to cell observables obtained by integrating over $Q$. For example, for a function $f:Q\times \Xi \to \R^m$, we can define a function $F:\Xi \to \R^m$ by integrating over the reference cell
\begin{equation}
    F(\omega) = \int_Q f(\bs,\omega)\dd\bs
\end{equation}
and then apply \cref{thm:ergodic} and \cref{cor:weighted_ergodic} to $F$. 

\begin{remark}
For two-scale functions $f(\bx,\bs,\omega)$, the ergodic theorem is applied for each fixed $\bx\in\overline\Omega$. The resulting almost-sure statements may depend on $\bx$, because the null set may depend on $\bx$. Whenever finitely or countably many almost-sure statements are used simultaneously, we work on their intersection, which still has probability one. We use this convention below without introducing separate notation for each exceptional set.
\end{remark}

\begin{remark}
\label{rem:derived_stationary_fields}
If $f(\bs,\omega)$ and $g(\bs,\omega)$ are measurable, then
\begin{equation*}
    f(\bs,T_{\bk}\omega)
    g(\bs,T_{\bk}\omega)
\end{equation*}
is lattice stationary. A convolution also preserves lattice stationarity whenever it is well defined. 
\end{remark}

\section{Assumptions and key consequences}
\label{sec:assumptions}

In this section, we collect key assumptions on the random charge density field, the microscopic potential, and the finite-volume interactions. Following the list of assumptions, we state the key consequences that are used in the analysis of the continuum limit.

\subsection{Assumptions}

\begin{assumption}[\textbf{(A1)} Continuity and cell neutrality]
\label{asmp:neutrality}
The map $\bx\mapsto\bar\rho(\bx,\cdot,\cdot)$ is continuous from
$\overline\Omega$ into $L^2(Q\times\Xi)$. Moreover,
\begin{equation}
    \int_Q\bar\rho(\bx,\bs,\omega)\dd\bs=0
    \label{eq:cell_neutrality}
\end{equation}
for every $\bx\in\overline\Omega$ and almost every $\omega\in\Xi$.
\end{assumption}

Since $\overline\Omega$ is compact, \textbf{A1} implies
\begin{equation}
    \sup_{\bx\in\overline\Omega}
    \E\!\left[
    \int_Q|\bar\rho(\bx,\bs,\cdot)|^2\dd\bs
    \right]<\infty.
    \label{eq:rho_moment}
\end{equation}
Taking the expectation in \cref{eq:cell_neutrality} gives
\begin{equation*}
    \int_Q\E[\bar\rho(\bx,\bs,\cdot)]\dd\bs=0.
\end{equation*}
However, the assumption does not require $\E[\bar\rho(\bx,\bs,\cdot)]$ to vanish pointwise in $\bs$. Hence the mean cell dipole, and therefore $\bp(\bx)$ in \cref{eq:effective_polarization}, can be nonzero. Neutrality only in expectation would permit random monopoles in individual cells and would not give the dipole-order decay used in the present analysis.

\begin{assumption}[\textbf{(A2)} Microscopic Coulomb potential]
\label{asmp:stationary}
For every $\bx\in\overline\Omega$ and almost every $\omega\in\Xi$, the
Coulomb convolution
\begin{equation}
    h(\bx,\by,\omega)=G\ast\rho(\bx,\cdot,\omega)(\by)
    \label{eq:a2_convolution}
\end{equation}
exists and its reference-cell representation $\bar{h}$ is $L^2$-integrable in the sense that
\begin{equation}
    \bar h(\bx,\cdot,\cdot)\in L^2(Q\times\Xi).
    \label{eq:a2_h_l2}
\end{equation}
Moreover, the expectation of the microscopic energy density, 
\begin{equation}
    \bx\longmapsto
    \E\!\left[
    \int_Q
    \bar\rho(\bx,\bs,\cdot)\bar h(\bx,\bs,\cdot)\dd\bs
    \right],
    \label{eq:local_observable_continuity}
\end{equation}
is continuous on $\overline\Omega$.
\end{assumption}

By the Cauchy--Schwarz inequality,
\begin{equation}
    \E\!\left[
    \int_Q
    |\bar\rho(\bx,\bs,\cdot)\bar h(\bx,\bs,\cdot)|
    \dd\bs
    \right]<\infty,
    \label{eq:local_observable_integrability}
\end{equation}
so \cref{thm:ergodic} applies to the microscopic energy density.

As we will see in the analysis of the continuum limit, the energies depend on the finite-volume potential $h_R$ defined in \cref{eq:truncated_potential_at_origin}. However, the limit of the energy more naturally involves the full-space potential $h$ and we find that
\begin{equation*}
    \nabla_{\boldsymbol\xi}
    \left[
    \frac{
    h_R(\bx,R\boldsymbol\xi,\omega)
    -
    h(\bx,R\boldsymbol\xi,\omega)
    }{R}
    \right]
\end{equation*}
has the weak limit $\nabla H_Q[\bp(\bx)]$; see \cref{prop:weak_finite_volume_field}. For the analysis of the energy, we require the convergence above to $\nabla H_Q[\bp(\bx)]$ to be strong in $L^2$ so that we can pass to the limit in the energy. This is taken as the third assumption. We also add to this assumption the convergence of the finite-volume potential $h_R$ in the translated region $R(\bk+Q)$, which is used to describe the interaction between two distinct material points in the nonlocal energy. We note that \cref{prop:weak_finite_volume_field} proves the weak convergence based on \textbf{A1} and \textbf{A2}, but the strong convergence is assumed in \textbf{A3}. Thus, \cref{prop:weak_finite_volume_field} is intended to provide the intuitive justification for the assumption.

\begin{assumption}[\textbf{(A3)} Strong finite-volume field convergence]
\label{asmp:local_finite_volume}
Let $h$ and $h_R$ be defined by
\cref{eq:microscopic_potential,eq:truncated_potential_at_origin}.
For every fixed $\bx\in\overline\Omega$ and almost every
$\omega\in\Xi$, as $R=2n+1\to\infty$,
\begin{equation}
    \nabla_{\boldsymbol\xi}
    \left[
    \frac{
    h_R(\bx,R\boldsymbol\xi,\omega)
    -
    h(\bx,R\boldsymbol\xi,\omega)
    }{R}
    \right]
    \longrightarrow
    \nabla H_Q[\bp(\bx)](\boldsymbol\xi)
    \label{eq:strong_finite_volume_field}
\end{equation}
strongly in $L^2(Q;\R^3)$. In addition, for every fixed
$\bk\in\Z^3\setminus\{\bm0\}$,
\begin{equation}
    \nabla_{\boldsymbol\xi}
    \left[
    \frac{
    h_R(\bx,R\boldsymbol\xi,\omega)
    }{R}
    \right]
    \longrightarrow
    \nabla H_Q[\bp(\bx)](\boldsymbol\xi)
    \label{eq:strong_translated_finite_volume_field}
\end{equation}
strongly in $L^2(\bk+Q;\R^3)$.
\end{assumption}

\subsection{Key consequences}

\begin{proposition}[Deterministic effective polarization]
\label{prop:deterministic_polarization}
Under \textbf{A1}, for every fixed $\bx\in\overline\Omega$,
\begin{equation}
    \bp_R(\bx,\omega)\longrightarrow\bp(\bx)
    \label{eq:deterministic_polarization}
\end{equation}
almost surely in $\omega$ and in $L^2(\Xi;\R^3)$.
\end{proposition}

\begin{proof}
Because the lattice is countable and the action is measure preserving,
\cref{eq:cell_neutrality} holds simultaneously in all cells for almost every
realization. Using $Q_R=\bigcup_{\bk\in I_R}(\bk+Q)$,
\begin{equation}\label{eq:polarization_ergodic_average}
    \bp_R(\bx,\omega)
    =
    \frac{1}{R^3}
    \sum_{\bk\in I_R}
    \int_Q
    (\bk+\bs)\bar\rho(\bx,\bs,T_{\bk}\omega)\dd\bs 
    =
    \frac{1}{R^3}
    \sum_{\bk\in I_R}
    \int_Q
    \bs\,\bar\rho(\bx,\bs,T_{\bk}\omega)\dd\bs,
\end{equation}
where the term with $\bk$ vanishes by cell neutrality \textbf{A1}. The
cell-dipole observable belongs to $L^2(\Xi;\R^3)$ by \textbf{A1}. Applying
\cref{thm:ergodic} to $F(\omega) = \int_Q \bs\,\bar\rho(\bx,\bs,\omega)\dd\bs$ gives the stated almost-sure and
$L^2(\Xi;\R^3)$ convergence to $\bp(\bx)$.
\end{proof}

\begin{proposition}[Microscopic Poisson equation and field integrability]
\label{prop:microscopic_potential_properties}
Under \textbf{A1}--\textbf{A2}, $h=G\ast\rho$ satisfies
\begin{equation}
    -\Delta_{\by}h(\bx,\by,\omega)
    =
    \rho(\bx,\by,\omega)
    \qquad
    \text{in }\mathcal D'(\R^3)
    \label{eq:microscopic_poisson_problem}
\end{equation}
for every fixed $\bx$ and almost every $\omega$. Moreover,
\begin{equation}
    \E\!\left[
    \int_Q
    |\nabla_{\by}h(\bx,\bs,\cdot)|^2
    \dd\bs
    \right]
    <\infty.
    \label{eq:stationary_field_energy_bound}
\end{equation}
\end{proposition}

\begin{proposition}[Weak finite-volume field limits]
\label{prop:weak_finite_volume_field}
Let $h$ and $h_R$ be defined by
\cref{eq:microscopic_potential,eq:truncated_potential_at_origin}.
Under \textbf{A1}--\textbf{A2}, for every fixed
$\bx\in\overline\Omega$ and almost every $\omega\in\Xi$,
\begin{equation}
    \nabla_{\boldsymbol\xi}
    \left[
    \frac{
    h_R(\bx,R\boldsymbol\xi,\omega)
    }{R}
    \right]
    \rightharpoonup
    \nabla H_Q[\bp(\bx)](\boldsymbol\xi)
    \label{eq:weak_truncated_finite_volume_field}
\end{equation}
weakly in $L^2(\R^3;\R^3)$ as $R=2n+1\to\infty$. Moreover,
\begin{equation}
    \nabla_{\boldsymbol\xi}
    \left[
    \frac{
    h_R(\bx,R\boldsymbol\xi,\omega)
    -
    h(\bx,R\boldsymbol\xi,\omega)
    }{R}
    \right]
    \rightharpoonup
    \nabla H_Q[\bp(\bx)](\boldsymbol\xi)
    \label{eq:weak_finite_volume_field}
\end{equation}
weakly in $L^2(Q;\R^3)$.
\end{proposition}

\cref{prop:microscopic_potential_properties,prop:weak_finite_volume_field}
are proved in \cref{app:weak_potential}.
\cref{prop:weak_finite_volume_field} identifies the weak limits of the
finite-volume fields that enter the local and nonlocal energies, while
\textbf{A3} assumes the corresponding strong convergences.

\begin{remark}[Admissibility of \textbf{A2}]
For a deterministic periodic neutral charge density, \textbf{A2} follows directly from the Fourier representation, as shown in \cref{sec:periodic_example}. For random perturbations, the existence of the lattice-stationary convolution depends on the perturbation law; it is established for the model in \cref{sec:random_example}.
\end{remark}

\section{Two-scale limit of the electrostatic energy}
\label{sec:continuum_limit}

The two-scale limit is taken in two steps. First, $\ell/\veps\to0$ with $\veps$ fixed. As mentioned in \cref{sec:setup}, the sequence is $\ell=\veps/(2n+1)$ and $R=\Repsl=\veps/\ell=2n+1$ throughout the article. Second, $\veps\to0$. The limits below are understood in this order: for each fixed $\veps$, the inner almost-sure limit is first identified with a deterministic quantity, and the ordinary limit $\veps\to0$ is then taken for these deterministic quantities. The following theorem is the main result of this work, and its proof is given in the two subsections below.

\begin{theorem}[Continuum limit of the electrostatic energy]
\label{thm:main}
Let \textbf{A1}--\textbf{A3} hold. Let $\bp_0$ denote the extension of
$\bp$ by zero outside $\Omega$, and let
$\phi\in\dot H^1(\R^3)$, unique up to an additive constant, satisfy
\begin{equation}
    \int_{\R^3}\nabla\phi\cdot\nabla\psi\dd\bx
    =
    \int_{\R^3}\bp_0\cdot\nabla\psi\dd\bx
    \qquad
    \text{for every }\psi\in\dot H^1(\R^3).
    \label{eq:macroscopic_poisson_solution}
\end{equation}
Then, for almost every $\omega\in\Xi$,
\begin{equation}
    \lim_{\veps\downarrow0}
    \lim_{\ell/\veps\downarrow0}
    \cE_{\veps,\ell}^{\mathrm{loc}}(\omega)
    =
    \frac12
    \int_\Omega
    \E\!\left[
    \int_Q
    \bar\rho(\bx,\bs,\cdot)\bar h(\bx,\bs,\cdot)
    \dd\bs
    \right]\dd\bx
    +
    \frac16
    \int_\Omega|\bp|^2\dd\bx,
    \label{eq:limiting_local_energy}
\end{equation}
and
\begin{equation}
    \lim_{\veps\downarrow0}
    \lim_{\ell/\veps\downarrow0}
    \cE_{\veps,\ell}^{\mathrm{nloc}}(\omega)
    =
    \frac12
    \int_{\R^3}|\nabla\phi|^2\dd\bx
    -
    \frac16
    \int_\Omega|\bp|^2\dd\bx.
    \label{eq:limiting_nonlocal_energy}
\end{equation}
Consequently,
\begin{equation}
    \lim_{\veps\downarrow0}
    \lim_{\ell/\veps\downarrow0}
    \cE_{\veps,\ell}(\omega)
    ={}
    \frac12
    \int_\Omega
    \E\!\left[
    \int_Q
    \bar\rho(\bx,\bs,\cdot)\bar h(\bx,\bs,\cdot)
    \dd\bs
    \right]\dd\bx +
    \frac12
    \int_{\R^3}|\nabla\phi|^2\dd\bx.
\label{eq:main_energy_limit}
\end{equation}
All three limiting energies are deterministic.
\end{theorem}

\begin{proposition}[Dipole-kernel representation of the nonlocal energy]
\label{prop:dipole_representation}
The nonlocal energy in \cref{eq:limiting_nonlocal_energy} satisfies
\begin{equation}
\begin{aligned}
    \frac12
    \int_{\R^3}|\nabla\phi|^2\dd\bx
    -
    \frac16
    \int_\Omega|\bp|^2\dd\bx
    =
    \frac12
    \pv
    \int_\Omega\int_\Omega
    \bK(\bx-\bx')
    :
    \bp(\bx)\otimes\bp(\bx')
    \dd\bx'\dd\bx.
\end{aligned}
\label{eq:nonlocal_dipole_representation}
\end{equation}
\end{proposition}

The proofs of \cref{thm:main} and \cref{prop:dipole_representation} are given in the next two subsections. 

\subsection{Analysis of the local energy}
\label{sec:local_limit}

Changing variables $\br=\ell\by$ and $\br'=\ell\bz$ in
\cref{eq:finite_local_energy}, and using
$G(\ell(\by-\bz))=\ell^{-1}G(\by-\bz)$, gives
\begin{equation}
    \cE_{\veps,\ell}^{\mathrm{loc}}(\omega)
    =
    \sum_{\bx\in\Omega_\veps}
    \veps^3 e_R^{\mathrm{loc}}(\bx,\omega),
    \qquad R=\frac{\veps}{\ell},
    \label{eq:local_riemann_sum}
\end{equation}
where
\begin{equation}
\begin{aligned}
    e_R^{\mathrm{loc}}(\bx,\omega)
    &:=
    \frac{1}{2R^3}
    \int_{Q_R}\int_{Q_R}
    G(\by-\bz)
    \rho(\bx,\by,\omega)
    \rho(\bx,\bz,\omega)
    \dd\bz\dd\by\\
    &=
    \frac{1}{2R^3}
    \int_{Q_R}
    \rho(\bx,\by,\omega)h_R(\bx,\by,\omega)\dd\by.
\end{aligned}
\label{eq:finite_local_density}
\end{equation}
Using $h_R = h + (h_R - h)$, we decompose the local energy density into two terms:
\begin{equation}
\begin{aligned}
    e_R^{\mathrm{loc}}(\bx,\omega)
    ={}&
    \frac{1}{2R^3}
    \int_{Q_R}\rho(\bx,\by,\omega)h(\bx,\by,\omega)\dd\by\\
    &+
    \frac{1}{2R^3}
    \int_{Q_R}
    \rho(\bx,\by,\omega)
    [h_R(\bx,\by,\omega)-h(\bx,\by,\omega)]
    \dd\by.
\end{aligned}
\label{eq:local_density_decomposition}
\end{equation}

Using the fact that $Q_R$ is the union of the cells $\bk+Q$ for $\bk\in I_R$, we can rewrite the first term in \cref{eq:local_density_decomposition} as a sum of stationary cell averages:
\begin{equation}
\begin{aligned}
    \frac{1}{R^3}\int_{Q_R}\rho h\dd\by
    =
    \frac{1}{R^3}
    \sum_{\bk\in I_R}
    \int_Q
    \bar\rho(\bx,\bs,T_{\bk}\omega)
    \bar h(\bx,\bs,T_{\bk}\omega)\dd\bs,
\end{aligned}
\label{eq:stationary_local_observable}
\end{equation}
where we used the relations between $\rho$ and $\bar\rho$ in \cref{eq:stationary_representation}, and between $h$ and $\bar h$ in \cref{eq:reference_cell_potential}. By \cref{thm:ergodic}, we have
\begin{equation}
    \frac{1}{R^3}\int_{Q_R}\rho h\dd\by
    \longrightarrow
    \E\!\left[
    \int_Q
    \bar\rho(\bx,\bs,\cdot)\bar h(\bx,\bs,\cdot)\dd\bs
    \right].
    \label{eq:stationary_local_limit}
\end{equation}

To treat the second term in \cref{eq:local_density_decomposition}, we need the following two lemmas. 

\begin{lemma}[Cell polarization field]
\label{lem:cell_polarization_field}
Under \textbf{A1}, for every fixed $\bx$ and almost every $\omega$, there exists
$\bar\bP(\bx,\cdot,\omega)\in L^2(Q;\R^3)$ such that
\begin{equation}
    -\nabla_{\bs}\cdot\bar\bP
    =
    \bar\rho
    \quad\text{in }Q,
    \qquad
    \bar\bP\cdot\bn_Q=0
    \quad\text{on }\partial Q,
    \label{eq:cell_polarization_problem}
\end{equation}
and
\begin{equation}
    \int_Q\bar\bP(\bx,\bs,\omega)\dd\bs
    =
    \int_Q\bs\,\bar\rho(\bx,\bs,\omega)\dd\bs.
    \label{eq:cell_polarization_mean}
\end{equation}
Moreover,
\begin{equation}
    \|\bar\bP(\bx,\cdot,\omega)\|_{L^2(Q)}
    \leq
    C\|\bar\rho(\bx,\cdot,\omega)\|_{L^2(Q)}.
    \label{eq:cell_polarization_estimate}
\end{equation}
\end{lemma}

\begin{proof}
For fixed $\bx$ and $\omega$, let $\psi\in H^1(Q)$ be the zero-mean solution of
the Neumann problem
\begin{equation*}
    -\Delta_{\bs}\psi=\bar\rho,
    \qquad
    \nabla_{\bs}\psi\cdot\bn_Q=0,
    \qquad
    \int_Q\psi\dd\bs=0.
\end{equation*}
Such a solution exists because $\int_Q\bar\rho\dd\bs=0$ by \textbf{A1}, which is the solvability condition for the Neumann problem.
We set $\bar\bP=\nabla_{\bs}\psi$. Then \cref{eq:cell_polarization_problem} follows directly, and the standard
estimate for the Neumann problem ($\|\bar\bP\|_{L^2} \leq \|\psi\|_{H^1} \leq C \|\bar\rho\|_{L^2}$) gives \cref{eq:cell_polarization_estimate}. Finally, integration by parts gives
\begin{equation*}
    \int_Q s_i\bar\rho\dd\bs
    =
    -\int_Q s_i\nabla_{\bs}\cdot\bar\bP\dd\bs
    =
    \int_Q\bar P_i\dd\bs,
\end{equation*}
which proves \cref{eq:cell_polarization_mean}.
\end{proof}

\begin{lemma}[Extended polarization field and its weak limit]
\label{lem:cell_polarization_weak_limit}
Under \textbf{A1}, let $\bar\bP$ be the cell polarization field from
\cref{lem:cell_polarization_field}, and define its cellwise extension by
\begin{equation}
    \bP(\bx,\bk+\bs,\omega)
    :=
    \bar\bP(\bx,\bs,T_{\bk}\omega),
    \qquad
    \bk\in\Z^3,\quad \bs\in Q.
    \label{eq:cell_polarization_extension}
\end{equation}
Then, for every fixed $\bx$ and almost every $\omega$,
\begin{equation}
    \bP(\bx,\cdot,\omega)
    \in
    H_{\mathrm{loc}}(\operatorname{div};\R^3),
\end{equation}
and
\begin{equation}
    -\nabla_{\by}\cdot\bP(\bx,\by,\omega)
    =
    \rho(\bx,\by,\omega)
    \qquad
    \text{in }\mathcal D'(\R^3).
    \label{eq:global_polarization_divergence}
\end{equation}
Moreover,
\begin{equation}
    \bP(\bx,R\,\cdot,\omega)
    \rightharpoonup
    \bp(\bx)
    \qquad
    \text{weakly in }L^2(Q;\R^3)
    \label{eq:cell_polarization_weak_limit}
\end{equation}
as $R=2n+1\to\infty$.
\end{lemma}

\begin{proof}
Since $\bar\bP\cdot\bn_Q=0$ on $\partial Q$, the normal component of the
cellwise extension vanishes on every cell boundary. Thus no interface terms
arise in the distributional divergence, and, for fixed $\bx$ and almost every $\omega$,
\begin{equation*}
    -\nabla_{\by}\cdot\bP(\bx,\by,\omega)
    =
    \rho(\bx,\by,\omega)
    \qquad
    \text{in }\mathcal D'(\R^3).
\end{equation*}
Together with the local $L^2$ integrability of $\bP$ and $\rho$, this gives $\bP(\bx, \cdot, \omega)\in H_{\mathrm{loc}}(\operatorname{div};\R^3)$.

To prove \cref{eq:cell_polarization_weak_limit}, let
$\boldsymbol\varphi\in C^1(\overline Q;\R^3)$. Since
$Q_R=\bigcup_{\bk\in I_R}(\bk+Q)$,
\begin{equation*}
    \int_Q
    \bP(\bx,R\boldsymbol\xi,\omega)
    \cdot\boldsymbol\varphi(\boldsymbol\xi)
    \dd\boldsymbol\xi
    =
    \frac{1}{R^3}
    \sum_{\bk\in I_R}
    \int_Q
    \bar\bP(\bx,\bs,T_{\bk}\omega)
    \cdot
    \boldsymbol\varphi\left(\frac{\bk+\bs}{R}\right)
    \dd\bs.
\end{equation*}
Since $\boldsymbol\varphi\in C^1(\overline Q;\R^3)$ and $Q$ is bounded,
\begin{equation*}
    \left|
    \boldsymbol\varphi\left(\frac{\bk+\bs}{R}\right)
    -
    \boldsymbol\varphi\left(\frac{\bk}{R}\right)
    \right|
    \leq
    \frac{C}{R}
\end{equation*}
uniformly in $\bk\in I_R$ and $\bs\in Q$. Hence
\begin{equation}
    \int_Q
    \bP(\bx,R\boldsymbol\xi,\omega)
    \cdot\boldsymbol\varphi(\boldsymbol\xi)
    \dd\boldsymbol\xi
    =
    \frac{1}{R^3}
    \sum_{\bk\in I_R}
    \left[
    \int_Q
    \bar\bP(\bx,\bs,T_{\bk}\omega)\dd\bs
    \right]
    \cdot
    \boldsymbol\varphi\left(\frac{\bk}{R}\right)
    +
    r_R,
\label{eq:cell_polarization_weak_limit_aux}
\end{equation}
where
\begin{equation*}
    |r_R|
    \leq
    \frac{C}{R}
    \frac{1}{R^3}
    \sum_{\bk\in I_R}
    \int_Q
    |\bar\bP(\bx,\bs,T_{\bk}\omega)|
    \dd\bs.
\end{equation*}
By \cref{eq:cell_polarization_estimate},
\begin{equation*}
    \E\!\left[
    \int_Q
    |\bar\bP(\bx,\bs,\cdot)|
    \dd\bs
    \right]
    <\infty,
\end{equation*}
and therefore the ergodic theorem implies that the average in the last
estimate remains bounded almost surely. Thus
\begin{equation*}
    r_R\longrightarrow0.
\end{equation*}

Moreover, by \cref{eq:cell_polarization_mean},
\begin{equation*}
    \bp(\bx)
    =
    \E\!\left[
    \int_Q
    \bar\rho(\bx,\bs,\cdot)\bs
    \dd\bs
    \right]
    =
    \E\!\left[
    \int_Q
    \bar\bP(\bx,\bs,\cdot)
    \dd\bs
    \right].
\end{equation*}
Hence the weighted ergodic theorem
\cref{cor:weighted_ergodic} gives
\begin{equation*}
    \frac{1}{R^3}
    \sum_{\bk\in I_R}
    \left[
    \int_Q
    \bar\bP(\bx,\bs,T_{\bk}\omega)\dd\bs
    \right]
    \cdot
    \boldsymbol\varphi\left(\frac{\bk}{R}\right) 
    \qquad\longrightarrow
    \bp(\bx)\cdot
    \int_Q
    \boldsymbol\varphi(\boldsymbol\xi)
    \dd\boldsymbol\xi.
\end{equation*}
Combining this with \cref{eq:cell_polarization_weak_limit_aux} proves the desired convergence for smooth test fields.

Finally,
\begin{equation*}
\begin{aligned}
    \|\bP(\bx,R\,\cdot,\omega)\|_{L^2(Q)}^2
    &=
    \frac{1}{R^3}
    \sum_{\bk\in I_R}
    \int_Q
    |\bar\bP(\bx,\bs,T_{\bk}\omega)|^2
    \dd\bs
\end{aligned}
\end{equation*}
is bounded almost surely by \cref{eq:cell_polarization_estimate,thm:ergodic}. Density of $C^1(\overline Q;\R^3)$ in $L^2(Q;\R^3)$ then yields \cref{eq:cell_polarization_weak_limit}.
\end{proof}

Continuing with the limit of the second term in \cref{eq:local_density_decomposition}, we define the scaled potential correction
\begin{equation}
    U_R(\bx,\boldsymbol\xi,\omega)
    :=
    \frac{
    h_R(\bx,R\boldsymbol\xi,\omega)
    -
    h(\bx,R\boldsymbol\xi,\omega)
    }{R}.
    \label{eq:scaled_potential_correction}
\end{equation}
The change of variables $\by=R\boldsymbol\xi$ and
\cref{eq:global_polarization_divergence} give
\begin{equation}
    \frac{1}{R^3}
    \int_{Q_R}
    \rho(\bx,\by,\omega)
    [h_R(\bx,\by,\omega)-h(\bx,\by,\omega)]
    \dd\by
    =
    \int_Q
    \bP(\bx,R\boldsymbol\xi,\omega)\cdot
    \nabla_{\boldsymbol\xi}U_R(\bx,\boldsymbol\xi,\omega)
    \dd\boldsymbol\xi.
\label{eq:local_correction_pairing}
\end{equation}
Note that there is no boundary term because $Q_R$ is a union of complete cells and the normal trace of $\bP$ vanishes on every cell boundary.

The first factor in \cref{eq:local_correction_pairing} converges weakly by
\cref{eq:cell_polarization_weak_limit} in \cref{lem:cell_polarization_weak_limit}, while the second converges strongly by the assumption \textbf{A3}. Therefore,
\begin{equation}
    \frac{1}{R^3}
    \int_{Q_R}
    \rho(\bx,\by,\omega)
    [h_R(\bx,\by,\omega)-h(\bx,\by,\omega)]
    \dd\by
    \longrightarrow
    \int_Q
    \bp(\bx)\cdot\nabla H_Q[\bp(\bx)]
    \dd\boldsymbol\xi.
\label{eq:local_depolarization_limit}
\end{equation}
For the unit cube, \cref{app:electrostatic_identities} shows that 
\begin{equation}
    \int_Q
    \bp\cdot\nabla H_Q[\bp]\dd\boldsymbol\xi
    =
    \frac13|\bp|^2.
    \label{eq:cube_depolarization_identity}
\end{equation}
Hence, combining \cref{eq:stationary_local_limit}, \cref{eq:local_depolarization_limit}, and the identity above, we have, for every fixed $\bx\in\overline\Omega$ and almost every $\omega$,
\begin{equation}
    e_R^{\mathrm{loc}}(\bx,\omega)
    \; \xrightarrow[R = \veps/\ell \to \infty]{} \;
    e_{\mathrm{loc}}(\bx)
    :=
    \frac12
    \E\!\left[
    \int_Q
    \bar\rho(\bx,\bs,\cdot)\bar h(\bx,\bs,\cdot)\dd\bs
    \right]
    +
    \frac16|\bp(\bx)|^2.
    \label{eq:local_density_convergence}
\end{equation}

\begin{proposition}[Limit of the local energy]
\label{prop:local_energy}
Under \textbf{A1}--\textbf{A3}, for almost every $\omega$,
\begin{equation}
    \lim_{\veps\downarrow0}
    \lim_{\ell/\veps\downarrow0}
    \cE_{\veps,\ell}^{\mathrm{loc}}(\omega)
    =
    \frac12
    \int_\Omega
    \E\!\left[
    \int_Q
    \bar\rho(\bx,\bs,\cdot)\bar h(\bx,\bs,\cdot)\dd\bs
    \right]\dd\bx
    +
    \frac16
    \int_\Omega|\bp(\bx)|^2\dd\bx.
    \label{eq:local_energy_limit}
\end{equation}
\end{proposition}

\begin{proof}
Fix $\veps>0$. The set $\Omega_{\veps}$ is finite. Hence, intersecting the almost-sure sets on which \cref{eq:local_density_convergence} holds for $\bx\in\Omega_{\veps}$ gives a set of probability one on which
\begin{equation*}
    \lim_{\ell/\veps\downarrow0}
    \cE_{\veps,\ell}^{\mathrm{loc}}(\omega)
    =
    \sum_{\bx\in\Omega_{\veps}}
    \veps^3 e_{\mathrm{loc}}(\bx).
\end{equation*}

By \textbf{A1}--\textbf{A2}, $e_{\mathrm{loc}}$ is continuous on
$\overline\Omega$, and hence uniformly continuous and bounded. Let
\begin{equation*}
    \Omega_{\veps}^{\mathrm{int}}
    :=
    \bigcup_{\bx\in\Omega_\veps}(\bx+\Qeps).
\end{equation*}
Since $|\veps Q|=\veps^3$,
\begin{equation*}
    \left|
    \sum_{\bx\in\Omega_\veps}
    \veps^3e_{\mathrm{loc}}(\bx)
    -
    \int_{\Omega_\veps^{\mathrm{int}}}
    e_{\mathrm{loc}}(\bz)\dd\bz
    \right| \leq
    |\Omega|
    \sup_{\substack{\bx,\bz\in\overline\Omega\\
    |\bz-\bx|\leq\sqrt3\,\veps/2}}
    |e_{\mathrm{loc}}(\bz)-e_{\mathrm{loc}}(\bx)|
    \longrightarrow0.
\end{equation*}
Moreover, since $\Omega$ is Lipschitz,
\begin{equation*}
    |\Omega\setminus\Omega_\veps^{\mathrm{int}}|
    \longrightarrow0,
\end{equation*}
and hence
\begin{equation*}
    \left|
    \int_{\Omega_\veps^{\mathrm{int}}}e_{\mathrm{loc}}\dd\bz
    -
    \int_\Omega e_{\mathrm{loc}}\dd\bz
    \right|
    \leq
    \|e_{\mathrm{loc}}\|_{L^\infty(\Omega)}
    |\Omega\setminus\Omega_\veps^{\mathrm{int}}|
    \longrightarrow0.
\end{equation*}
Therefore,
\begin{equation*}
    \lim_{\veps\downarrow0} \lim_{\ell/\veps\downarrow0}
    \cE_{\veps,\ell}^{\mathrm{loc}}(\omega) = \int_{\Omega} e_{\mathrm{loc}}(\bz) \dd \bz.
\end{equation*}
Substituting the definition of $e_{\mathrm{loc}}$ from \cref{eq:local_density_convergence} gives the desired result.
\end{proof}


\subsection{Analysis of the nonlocal energy}
\label{sec:nonlocal_limit}

For fixed $\veps$, the first limit replaces the microscopic interaction
between two material points by the interaction between uniformly
polarized cells. This follows from \textbf{A3} and the weak convergence
of the cell polarization field. The limit $\veps\downarrow0$ is then
obtained by relating the resulting polarized-cell energy to the
electrostatic field generated by the effective polarization $\bp$.


\begin{proposition}[First-scale limit of the nonlocal energy]
\label{prop:first_scale_reduction}
For fixed $\veps$, define
\begin{equation}
    \cE_\veps^{\mathrm{cell}}
    :=
    \frac12
    \sum_{\substack{\bx,\bx'\in\Omega_\veps\\\bx\ne\bx'}}
    \veps^3
    B_{(\bx'-\bx)/\veps}
    \bigl(\bp(\bx),\bp(\bx')\bigr),
    \label{eq:polarized_cell_energy}
\end{equation}
where $B_{\bk}$ is defined in \cref{eq:polarized_cell_pair_interaction}. Under \textbf{A1}--\textbf{A3}, for almost every $\omega\in\Xi$,
\begin{equation}
    \lim_{\ell/\veps\downarrow0}
    \cE_{\veps,\ell}^{\mathrm{nloc}}(\omega)
    =
    \cE_\veps^{\mathrm{cell}}.
    \label{eq:first_scale_cell_reduction}
\end{equation}
\end{proposition}

\begin{proof}
Fix $\veps$ and two distinct material points
$\bx,\bx'\in\Omega_\veps$, and write $\bx'=\bx+\veps\bk$ for $\bk\in\Z^3\setminus\{\bm0\}$. With
$\br=\ell\by$, $\br'=\ell\bz$, and $R=\veps/\ell$, the corresponding
term in \cref{eq:finite_nonlocal_energy} becomes
\begin{equation}
\begin{aligned}
    &\int_{\veps Q}\int_{\veps Q}
    \rho^\ell(\bx,\br,\omega)
    \rho^\ell(\bx',\br',\omega)
    G(\bx+\br-\bx'-\br')
    \dd\br'\dd\br\\
    &\qquad=
    \frac{\veps^3}{R^3}
    \int_{Q_R}\int_{Q_R}
    G(\by-\bz-R\bk)
    \rho(\bx,\by,\omega)
    \rho(\bx',\bz,\omega)
    \dd\by\dd\bz =
    \veps^3
    B_{R,\bk}(\bx,\bx',\omega),
\end{aligned}
\label{eq:pair_rescaling_for_cell_limit}
\end{equation}
where $B_{R,\bk}$ is defined in \cref{eq:finite_volume_pair_interaction}.

We next determine the limit of $B_{R,\bk}$. From the definition of $B_{R,\bk}$, we have
\begin{equation}
    B_{R,\bk}(\bx,\bx',\omega)
    =
    \frac{1}{R^3}
    \int_{Q_R}
    \rho(\bx',\bz,\omega)
    h_R(\bx,\bz+R\bk,\omega)
    \dd\bz.
    \label{eq:pair_interaction_hr}
\end{equation}
Let $\bP$ be the extended cell polarization field from
\cref{lem:cell_polarization_weak_limit}. Since
\begin{equation*}
    -\nabla_{\bz}\cdot\bP(\bx',\bz,\omega)
    =
    \rho(\bx',\bz,\omega)
\end{equation*}
and $\bP\cdot\bn=0$ on the boundaries of the microscopic cells,
integration by parts over $Q_R$ gives
\begin{equation*}
\begin{aligned}
    B_{R,\bk}(\bx,\bx',\omega)
    =
    \frac{1}{R^3}
    \int_{Q_R}
    \bP(\bx',\bz,\omega)
    \cdot
    \nabla_{\bz}
    h_R(\bx,\bz+R\bk,\omega)
    \dd\bz.
\end{aligned}
\end{equation*}
Setting $\bz=R\boldsymbol\xi$ and using the chain rule gives
\begin{equation}
\begin{aligned}
    B_{R,\bk}(\bx,\bx',\omega)
    =
    \int_Q
    \bP(\bx',R\boldsymbol\xi,\omega)
    \cdot
    \nabla_{\boldsymbol\xi}
    \left[
    \frac{
    h_R(\bx,R(\boldsymbol\xi+\bk),\omega)
    }{R}
    \right]
    \dd\boldsymbol\xi.
\end{aligned}
\label{eq:pair_interaction_field_form}
\end{equation}

By \cref{lem:cell_polarization_weak_limit},
\begin{equation}
    \bP(\bx',R\,\cdot,\omega)
    \rightharpoonup
    \bp(\bx')
    \qquad
    \text{weakly in }L^2(Q;\R^3).
    \label{eq:pair_polarization_weak_limit}
\end{equation}
The second part of \textbf{A3}, applied on $\bk+Q$, gives
\begin{equation}
    \nabla_{\boldsymbol\xi}
    \left[
    \frac{
    h_R(\bx,R(\boldsymbol\xi+\bk),\omega)
    }{R}
    \right]
    \longrightarrow
    \nabla H_Q[\bp(\bx)](\boldsymbol\xi+\bk)
    \label{eq:pair_field_strong_limit}
\end{equation}
strongly in $L^2(Q;\R^3)$. Therefore,
\cref{eq:pair_interaction_field_form} and the weak--strong convergence
give
\begin{equation}
\begin{aligned}
    B_{R,\bk}(\bx,\bx',\omega)
    \longrightarrow
    \int_Q
    \bp(\bx')
    \cdot
    \nabla H_Q[\bp(\bx)](\boldsymbol\xi+\bk)
    \dd\boldsymbol\xi.
\end{aligned}
\label{eq:pair_interaction_polarized_limit}
\end{equation}

The right-hand side is precisely the polarized-cell interaction.
Indeed, using the weak equation for
$H_{\bk+Q}[\bp(\bx')]$ in \cref{eq:polarized_cell_pair_interaction}, we have
\begin{equation*}
\begin{aligned}
    B_{\bk}
    \bigl(\bp(\bx),\bp(\bx')\bigr)
    &=
    \int_{\R^3}
    \nabla H_Q[\bp(\bx)]
    \cdot
    \nabla H_{\bk+Q}[\bp(\bx')]
    \dd\boldsymbol\xi =
    \int_{\bk+Q}
    \bp(\bx')
    \cdot
    \nabla H_Q[\bp(\bx)]
    \dd\boldsymbol\xi\\
    &=
    \int_Q
    \bp(\bx')
    \cdot
    \nabla H_Q[\bp(\bx)](\boldsymbol\xi+\bk)
    \dd\boldsymbol\xi.
\end{aligned}
\end{equation*}
Hence
\begin{equation}
    B_{R,\bk}(\bx,\bx',\omega)
    \longrightarrow
    B_{\bk}
    \bigl(\bp(\bx),\bp(\bx')\bigr).
    \label{eq:finite_volume_pair_limit}
\end{equation}

Substituting \cref{eq:pair_rescaling_for_cell_limit} into
\cref{eq:finite_nonlocal_energy} gives
\begin{equation*}
    \cE_{\veps,\ell}^{\mathrm{nloc}}(\omega)
    =
    \frac12
    \sum_{\substack{\bx,\bx'\in\Omega_\veps\\\bx\ne\bx'}}
    \veps^3
    B_{R,(\bx'-\bx)/\veps}(\bx,\bx',\omega).
\end{equation*}
For fixed $\veps$, this is a finite sum. Therefore
\cref{eq:finite_volume_pair_limit} may be applied simultaneously to all
pairs, and passing the limit through the sum gives
\cref{eq:first_scale_cell_reduction}.
\end{proof}


\begin{proposition}[Limit of the nonlocal energy]
\label{prop:nonlocal_energy}
Under \textbf{A1}--\textbf{A3}, for almost every $\omega\in\Xi$,
\begin{equation}
    \lim_{\veps\downarrow0}
    \lim_{\ell/\veps\downarrow0}
    \cE_{\veps,\ell}^{\mathrm{nloc}}(\omega)
    =
    \frac12
    \int_{\R^3}|\nabla\phi|^2\dd\bx
    -
    \frac16
    \int_\Omega|\bp(\bx)|^2\dd\bx,
    \label{eq:nonlocal_energy_limit}
\end{equation}
where $\phi$ is defined by
\cref{eq:macroscopic_poisson_solution}.
\end{proposition}

\begin{proof}
By \textbf{A1} and \cref{eq:effective_polarization}, $\bp$ is continuous
on $\overline\Omega$. In particular, its zero extension $\bp_0$ belongs
to $L^2(\R^3;\R^3)$. Therefore the weak problem
\cref{eq:macroscopic_poisson_solution} has a solution
$\phi\in\dot H^1(\R^3)$, unique up to an additive constant.

Let $\bp_\veps$ be the piecewise-constant polarization
\begin{equation}
    \bp_\veps(\bz)
    :=
    \sum_{\bx\in\Omega_\veps}
    \bp(\bx)\chi_{\bx+\veps Q}(\bz).
    \label{eq:piecewise_constant_polarization}
\end{equation}
Using $|\bz-\bx|\leq\sqrt3\,\veps/2$ for $\bz\in\bx+\veps Q$, we obtain
\begin{equation*}
    \|\bp_\veps-\bp_0\|_{L^2(\R^3)}^2
    \leq
    |\Omega|
    \sup_{\substack{\bx,\bz\in\overline\Omega\\
    |\bx-\bz|\leq\sqrt3\,\veps/2}}
    |\bp(\bx)-\bp(\bz)|^2 \;+\;
    \|\bp\|_{L^\infty(\Omega)}^2
    \left|
    \Omega\setminus
    \bigcup_{\bx\in\Omega_\veps}
    (\bx+\veps Q)
    \right|.
\end{equation*}
It follows that
\begin{equation}
    \bp_\veps
    \longrightarrow
    \bp_0
    \qquad
    \text{strongly in }L^2(\R^3;\R^3).
    \label{eq:strong_polarization_extension}
\end{equation}

Let $\phi_\veps\in\dot H^1(\R^3)$ satisfy
\begin{equation}
    \int_{\R^3}
    \nabla\phi_\veps\cdot\nabla\psi
    \dd\bz
    =
    \int_{\R^3}
    \bp_\veps\cdot\nabla\psi
    \dd\bz
    \qquad
    \text{for every }\psi\in\dot H^1(\R^3).
    \label{eq:piecewise_polarization_problem}
\end{equation}
Subtracting \cref{eq:macroscopic_poisson_solution} from
\cref{eq:piecewise_polarization_problem} and taking
$\psi=\phi_\veps-\phi$ gives
\begin{equation*}
    \|\nabla\phi_\veps-\nabla\phi\|_{L^2(\R^3)}^2
    =
    \int_{\R^3}
    (\bp_\veps-\bp_0)
    \cdot
    (\nabla\phi_\veps-\nabla\phi)
    \dd\bz \leq
    \|\bp_\veps-\bp_0\|_{L^2(\R^3)}
    \|\nabla\phi_\veps-\nabla\phi\|_{L^2(\R^3)}.
\end{equation*}
Therefore,
\begin{equation}
    \nabla\phi_\veps
    \longrightarrow
    \nabla\phi
    \qquad
    \text{strongly in }L^2(\R^3;\R^3).
    \label{eq:macroscopic_field_convergence}
\end{equation}

We next relate the field energy of $\phi_\veps$ to
$\cE_\veps^{\mathrm{cell}}$. By
\cref{eq:piecewise_constant_polarization,eq:piecewise_polarization_problem},
for every $\psi\in\dot H^1(\R^3)$,
\begin{equation}\label{eq:phi_psi_relation}
    \int_{\R^3}
    \nabla\phi_\veps\cdot\nabla\psi
    \dd\bz
    =
    \sum_{\bx\in\Omega_\veps}
    \int_{\bx+\veps Q}
    \bp(\bx)\cdot\nabla\psi
    \dd\bz.
\end{equation}
For fixed $\bx\in\Omega_\veps$, use the weak equation for
$H_Q[\bp(\bx)]$ with
\begin{equation*}
    \eta(\boldsymbol\xi)
    :=
    \psi(\bx+\veps\boldsymbol\xi).
\end{equation*}
Since
$\nabla_{\boldsymbol\xi}\eta
=
\veps\nabla\psi(\bx+\veps\boldsymbol\xi)$,
\begin{equation*}
    \int_{\R^3}
    \nabla H_Q[\bp(\bx)](\boldsymbol\xi)
    \cdot
    \nabla\psi(\bx+\veps\boldsymbol\xi)
    \dd\boldsymbol\xi
    =
    \int_Q
    \bp(\bx)\cdot
    \nabla\psi(\bx+\veps\boldsymbol\xi)
    \dd\boldsymbol\xi.
\end{equation*}
Changing variables
$\bz=\bx+\veps\boldsymbol\xi$ gives
\begin{equation*}
    \int_{\R^3}
    \nabla H_Q[\bp(\bx)]
    \left(
    \frac{\bz-\bx}{\veps}
    \right)
    \cdot\nabla\psi(\bz)
    \dd\bz
    =
    \int_{\bx+\veps Q}
    \bp(\bx)\cdot\nabla\psi(\bz)
    \dd\bz.
\end{equation*}
Summing over $\bx\in\Omega_\veps$ and comparing with
\cref{eq:phi_psi_relation} yields
\begin{equation}
    \nabla\phi_\veps(\bz)
    =
    \sum_{\bx\in\Omega_\veps}
    \nabla H_Q[\bp(\bx)]
    \left(
    \frac{\bz-\bx}{\veps}
    \right)
    \qquad
    \text{in }L^2(\R^3;\R^3).
    \label{eq:piecewise_polarization_field}
\end{equation}

Using the preceding representation of $\nabla\phi_\veps$ and expanding the square in $\frac12\int_{\R^3}|\nabla\phi_\veps|^2\dd\bz$, we obtain a double sum indexed by $\bx$ and $\bx'$. Its diagonal terms have $\bx=\bx'$, and its off-diagonal terms have $\bx\ne\bx'$. Set $\bk=(\bx'-\bx)/\veps$ and $\bz=\bx+\veps\boldsymbol\xi$. An off-diagonal term can then be written as
\begin{equation*}
\begin{aligned}
    &\int_{\R^3}
    \nabla H_Q[\bp(\bx)]
    \left(
    \frac{\bz-\bx}{\veps}
    \right)
    \cdot
    \nabla H_Q[\bp(\bx')]
    \left(
    \frac{\bz-\bx'}{\veps}
    \right)
    \dd\bz\\
    &\qquad=
    \veps^3
    \int_{\R^3}
    \nabla H_Q[\bp(\bx)](\boldsymbol\xi)
    \cdot
    \nabla H_Q[\bp(\bx')]
    (\boldsymbol\xi-\bk)
    \dd\boldsymbol\xi =
    \veps^3
    B_{\bk}
    \bigl(\bp(\bx),\bp(\bx')\bigr),
\end{aligned}
\end{equation*}
by \cref{eq:polarized_cell_pair_interaction}. For the diagonal terms $\bx = \bx'$, the same change of variables and \cref{eq:cube_depolarization_energy} give
\begin{equation*}
    \frac12
    \int_{\R^3}
    \left|
    \nabla H_Q[\bp(\bx)]
    \left(
    \frac{\bz-\bx}{\veps}
    \right)
    \right|^2
    \dd\bz =
    \frac{\veps^3}{2}
    \int_{\R^3}
    |\nabla H_Q[\bp(\bx)]|^2
    \dd\boldsymbol\xi
    =
    \frac{\veps^3}{6}
    |\bp(\bx)|^2.
\end{equation*}
Consequently,
\begin{equation}
\begin{aligned}
    \frac12
    \int_{\R^3}
    |\nabla\phi_\veps|^2
    \dd\bz
    &=
    \frac12
    \sum_{\substack{\bx,\bx'\in\Omega_\veps\\\bx\ne\bx'}}
    \veps^3
    B_{(\bx'-\bx)/\veps}
    \bigl(\bp(\bx),\bp(\bx')\bigr) +
    \frac16
    \sum_{\bx\in\Omega_\veps}
    \veps^3|\bp(\bx)|^2\\
    &=
    \cE_\veps^{\mathrm{cell}}
    +
    \frac16
    \sum_{\bx\in\Omega_\veps}
    \veps^3|\bp(\bx)|^2.
\end{aligned}
\label{eq:cell_energy_decomposition}
\end{equation}

By \cref{eq:macroscopic_field_convergence},
\begin{equation*}
    \frac12
    \int_{\R^3}
    |\nabla\phi_\veps|^2
    \dd\bz
    \longrightarrow
    \frac12
    \int_{\R^3}
    |\nabla\phi|^2
    \dd\bz.
\end{equation*}
Since $\bp$ is continuous on $\overline\Omega$, the Riemann-sum limit
and
$|\Omega\setminus\Omega_\veps^{\mathrm{int}}|\to0$ give
\begin{equation*}
    \sum_{\bx\in\Omega_\veps}
    \veps^3|\bp(\bx)|^2
    \longrightarrow
    \int_\Omega
    |\bp(\bx)|^2
    \dd\bx.
\end{equation*}
Combining this limit with \cref{eq:cell_energy_decomposition}, we have
\begin{equation}
    \lim_{\veps\downarrow0}
    \cE_\veps^{\mathrm{cell}}
    =
    \frac12
    \int_{\R^3}
    |\nabla\phi|^2
    \dd\bx
    -
    \frac16
    \int_\Omega
    |\bp(\bx)|^2
    \dd\bx.
    \label{eq:cell_nonlocal_limit}
\end{equation}
Together with
\cref{eq:first_scale_cell_reduction}, this proves
\cref{eq:nonlocal_energy_limit}.
\end{proof}


\begin{proof}[Proof of \cref{prop:dipole_representation}]
Consider the $\phi$ defined in \cref{thm:main}. We take $\psi=\phi$ in \cref{eq:macroscopic_poisson_solution} to get
\begin{equation}
    \int_{\R^3}|\nabla\phi|^2\dd\bx
    =
    \int_{\R^3}\bp_0\cdot\nabla\phi\dd\bx.
    \label{eq:macroscopic_field_energy_identity}
\end{equation}
The weak equation \cref{eq:macroscopic_poisson_solution} is equivalent to
\begin{equation}\label{eq:phi_poisson_equation}
    -\Delta\phi
    =
    -\nabla\cdot\bp_0
    \qquad
    \text{in }\mathcal D'(\R^3).
\end{equation}
Using the Fourier transform identities
\begin{equation*}
    \mathcal F[\nabla f](\boldsymbol\xi)
    =
    \mathrm i\boldsymbol\xi\,
    \mathcal F[f](\boldsymbol\xi),
    \qquad
    \mathcal F[-\Delta f](\boldsymbol\xi)
    =
    |\boldsymbol\xi|^2
    \mathcal F[f](\boldsymbol\xi)
\end{equation*}
in \cref{eq:phi_poisson_equation} gives
\begin{equation}
    \mathcal F[\phi](\boldsymbol\xi)
    =
    -\mathrm i\frac{\boldsymbol\xi}{|\boldsymbol\xi|^2}\cdot
    \mathcal F[\bp_0](\boldsymbol\xi), \qquad \mathcal F[\nabla\phi](\boldsymbol\xi)
    =
    \frac{
    \boldsymbol\xi\otimes\boldsymbol\xi
    }{
    |\boldsymbol\xi|^2
    }
    \mathcal F[\bp_0](\boldsymbol\xi),
    \qquad
    \boldsymbol\xi\ne\bm0.
    \label{eq:fourier_macroscopic_field}
\end{equation}
Since $-\Delta G=\delta_{\bm0}$, we have $|\boldsymbol\xi|^2 \mathcal F[G](\boldsymbol\xi) = 1$ and $\mathcal F[-\nabla^2G](\boldsymbol\xi) = \boldsymbol\xi\otimes\boldsymbol\xi / |\boldsymbol\xi|^2$. 
Comparing this with \cref{eq:fourier_macroscopic_field} gives
\begin{equation}
    \nabla\phi
    =
    (-\nabla^2G)\ast\bp_0
    \qquad
    \text{in }L^2(\R^3;\R^3).
    \label{eq:macroscopic_field_convolution}
\end{equation}

Using \cref{eq:dipole_distribution_identity}, we obtain
\begin{equation}
    \nabla\phi
    =
    \pv(\bK\ast\bp_0)
    +
    \frac13\bp_0
    \qquad
    \text{in }L^2(\R^3;\R^3).
    \label{eq:macroscopic_field_dipole_representation}
\end{equation}
The kernel $\bK$ is homogeneous of degree $-3$, smooth on the unit sphere, and has zero spherical mean. Hence it defines a Calderón--Zygmund singular integral operator on $\R^3$. In particular, the truncated operators
\begin{equation*}
    \int_{|\bx-\bx'|>\delta}
    \bK(\bx-\bx')
    \bp_0(\bx')
    \dd\bx'
\end{equation*}
converge in $L^2(\R^3;\R^3)$ as $\delta\downarrow0$, defining
$\pv(\bK\ast\bp_0)$; see \cite[Chapter~II]{Stein1970}. The same singular-integral estimate is used for the dipole kernel in \cite[Section~6]{JamesMuller1994}.

Substituting
\cref{eq:macroscopic_field_dipole_representation} into
\cref{eq:macroscopic_field_energy_identity} gives
\begin{equation*}
    \frac12\int_{\R^3}|\nabla\phi|^2\dd\bx
    =
    \frac12\pv\int_{\R^3}\int_{\R^3}
    \bK(\bx-\bx'):\bp_0(\bx)\otimes\bp_0(\bx')
    \dd\bx'\dd\bx
    +
    \frac16\int_{\R^3}|\bp_0|^2\dd\bx.
\end{equation*}
Since $\bp_0$ is the extension of $\bp$ by zero outside $\Omega$, we have
\begin{equation*}
\begin{aligned}
    \frac12
    \int_{\R^3}|\nabla\phi|^2\dd\bx
    -
    \frac16
    \int_\Omega|\bp|^2\dd\bx
    =
    \frac12
    \pv
    \int_\Omega\int_\Omega
    \bK(\bx-\bx')
    :
    \bp(\bx)\otimes\bp(\bx')
    \dd\bx'\dd\bx,
\end{aligned}
\end{equation*}
which proves \cref{eq:nonlocal_dipole_representation}.
\end{proof}

\subsection{Interpretations of the results}
\label{sec:interpretation}

The local and nonlocal limits retain different information about the
microscopic random field. The first term in \cref{eq:limiting_local_energy} is the ensemble average of the microscopic Coulomb energy and therefore retains the statistics of the charge density. The second local term is the depolarization energy associated with the finite cubic material point.

By contrast, the nonlocal limit depends only on the effective polarization $\bp$. Microscopic random fields with different local statistics but the same effective polarization therefore have the same continuum nonlocal energy. The cubic self-term in the potential representation of the nonlocal energy is $-\frac16\int_\Omega|\bp|^2$, which cancels the corresponding $+\frac16\int_\Omega|\bp|^2$ contribution from the local energy. The total continuum energy consequently consists of the ensemble-averaged microscopic cell energy and the macroscopic field energy. If $\bp=\bm0$, the limiting nonlocal energy vanishes, while the microscopic contribution to the local energy need not vanish.

This separation gives a precise limitation on continuum descriptions based only on polarization. Such descriptions determine the long-range field energy, including its dependence on bound charge and the specimen boundary, but do not determine the microscopic ensemble-averaged cell energy. The distinction is consistent with the role of depolarizing fields and polarization fluctuations in finite systems \cite{PonomarevaBellaicheResta2007} and with dipolar discrete-to-continuum limits in which lattice regularization produces local corrections in addition to the macroscopic field \cite{JamesMuller1994,MullerSchlomerkemper2002}. Here the local correction is derived from the underlying random charge density, and ergodic averaging makes both parts of the limiting energy deterministic.

\section{Examples}
\label{sec:examples}

We consider two special cases of the general framework. The first is a deterministic periodic charge density and connects the present result with existing continuum descriptions of periodic ionic crystals. The second is a lattice-stationary model in which the charges undergo independent cellwise random displacements about their reference positions.

\subsection{Periodic charge density: deterministic case}
\label{sec:periodic_example}

Let $\Xi=\{\omega_0\}$ with $\Pp(\{\omega_0\})=1$ and
$T_{\bk}\omega_0=\omega_0$ for every $\bk\in\Z^3$. Let
$\rho_{\mathrm{per}}:\overline\Omega\times Q\to\R$ be such that
$\bx\mapsto\rho_{\mathrm{per}}(\bx,\cdot)$ is continuous in $L^2(Q)$,
and assume that the profiles are uniformly bounded and supported in a
fixed compact subset of $Q$. Set
\begin{equation}
    \bar\rho(\bx,\bs,\omega_0)
    :=
    \rho_{\mathrm{per}}(\bx,\bs),
    \qquad
    \int_Q
    \rho_{\mathrm{per}}(\bx,\bs)\dd\bs
    =
    0.
    \label{eq:periodic_rho_bar}
\end{equation}
Then
\begin{equation}
    \rho(\bx,\bk+\bs,\omega_0)
    =
    \rho_{\mathrm{per}}(\bx,\bs),
    \label{eq:periodic_stationary_field}
\end{equation}
so \textbf{A1} follows immediately. The effective polarization is
\begin{equation}
    \bp(\bx)
    =
    \int_Q
    \rho_{\mathrm{per}}(\bx,\bs)\bs\dd\bs.
    \label{eq:periodic_cell_polarization}
\end{equation}

We next write the Fourier series of the periodic charge density as
\begin{equation*}
    \rho_{\mathrm{per}}(\bx,\bs)
    =
    \sum_{\bk\in\Z^3}
    \widehat\rho_{\bk}(\bx)
    e^{2\pi i\bk\cdot\bs}.
\end{equation*}
Neutrality gives
$\widehat\rho_{\bm0}(\bx)=0$. The lattice-stationary Coulomb potential in this periodic setting is the zero-mean periodic solution
\begin{equation}
    h_{\mathrm{per}}(\bx,\bs)
    =
    \sum_{\bk\in\Z^3\setminus\{\bm0\}}
    \frac{\widehat\rho_{\bk}(\bx)}
    {|2\pi\bk|^2}
    e^{2\pi i\bk\cdot\bs},
    \label{eq:periodic_coulomb_convolution}
\end{equation}
which satisfies
$-\Delta_{\bs}h_{\mathrm{per}}=\rho_{\mathrm{per}}$ and gives
$\bar h(\bx,\bs,\omega_0)=h_{\mathrm{per}}(\bx,\bs)$. Equivalently, this gives the distributional Coulomb convolution
$G\ast\rho_{\mathrm{per}}$ with the zero Fourier mode fixed to zero.

Since $|\bk|\geq1$ for
$\bk\in\Z^3\setminus\{\bm0\}$, Parseval's identity gives
\begin{equation}
    \|h_{\mathrm{per}}(\bx,\cdot)\|_{L^2(Q)}^2
    =
    \sum_{\bk\in\Z^3\setminus\{\bm0\}}
    \frac{|\widehat\rho_{\bk}(\bx)|^2}
    {|2\pi\bk|^4}
    \leq
    \frac{1}{(2\pi)^4}
    \|\rho_{\mathrm{per}}(\bx,\cdot)\|_{L^2(Q)}^2.
\label{eq:periodic_potential_l2_estimate}
\end{equation}
Because $\Xi$ consists of a single point,
$L^2(Q\times\Xi)=L^2(Q)$. Moreover, applying the same estimate to the
difference at two material points gives
\begin{equation}
    \|h_{\mathrm{per}}(\bx,\cdot)
    -
    h_{\mathrm{per}}(\bx',\cdot)\|_{L^2(Q)}
    \leq
    \frac{1}{(2\pi)^2}
    \|\rho_{\mathrm{per}}(\bx,\cdot)
    -
    \rho_{\mathrm{per}}(\bx',\cdot)\|_{L^2(Q)}.
\label{eq:periodic_potential_continuity}
\end{equation}
Thus $h_{\mathrm{per}}$ is continuous in $\bx$ as an $L^2(Q)$-valued
field. Together with the Cauchy--Schwarz inequality,
\cref{eq:periodic_potential_l2_estimate,eq:periodic_potential_continuity}
also imply continuity of
\begin{equation*}
    \bx
    \longmapsto
    \int_Q
    \rho_{\mathrm{per}}(\bx,\bs)
    h_{\mathrm{per}}(\bx,\bs)
    \dd\bs.
\end{equation*}
Hence \textbf{A2} holds.

It remains to verify the strong finite-volume limits in \textbf{A3}.
For a single periodic cell, compare its electric field with the field
$\nabla H_Q[\bp(\bx)]$ generated by a uniformly polarized cell having
the same dipole moment. The difference has zero total charge and zero
dipole moment, and therefore its field decays as $|\by|^{-4}$.
Consequently, its lattice translates are absolutely summable. This gives
\begin{equation}
    \nabla_{\boldsymbol\xi}
    \left[
    \frac{
    h_R(\bx,R\boldsymbol\xi,\omega_0)
    -
    h_{\mathrm{per}}(\bx,R\boldsymbol\xi)
    }{R}
    \right]
    \longrightarrow
    \nabla H_Q[\bp(\bx)](\boldsymbol\xi)
    \label{eq:periodic_A3_local}
\end{equation}
strongly in $L^2(Q;\R^3)$ and, for every fixed
$\bk\in\Z^3\setminus\{\bm0\}$,
\begin{equation}
    \nabla_{\boldsymbol\xi}
    \left[
    \frac{
    h_R(\bx,R\boldsymbol\xi,\omega_0)
    }{R}
    \right]
    \longrightarrow
    \nabla H_Q[\bp(\bx)](\boldsymbol\xi)
    \label{eq:periodic_A3_nonlocal}
\end{equation}
strongly in $L^2(\bk+Q;\R^3)$. The estimates are given in
\cref{app:periodic_finite_volume_fields}. Thus
\textbf{A1}--\textbf{A3} hold for this periodic class.

The local energy density is therefore
\begin{equation}
    e_{\mathrm{loc}}^{\mathrm{per}}(\bx)
    =
    \frac12
    \int_Q
    \rho_{\mathrm{per}}(\bx,\bs)
    h_{\mathrm{per}}(\bx,\bs)
    \dd\bs
    +
    \frac16
    |\bp(\bx)|^2.
    \label{eq:periodic_local_energy}
\end{equation}
For this periodic class, \cref{thm:main} establishes convergence to the continuum energy obtained formally in \cite{MarshallDayal2014}. The recent two-scale analysis in \cite{SenWangBreitzmanDayal2026} further shows, for periodic dielectric crystals, how the bulk polarization associated with a chosen unit cell is accompanied by the corresponding boundary charge. Our fixed reference-cell
convention selects one such representation.

\subsection{Random perturbations of charges from their reference positions}
\label{sec:random_example}

We next consider an analytically tractable model of independent cellwise random displacements. Independently perturbed lattices provide a standard stationary and ergodic stochastic-lattice construction \cite{BlancLeBrisLions2006Discrete,BlancLeBrisLions2007Energy,Gabrielli2004}. Our model assigns a neutral collection of charge profiles to each cell and allows the reference positions, profile shapes, and displacement statistics to depend on the macroscopic variable $\bx$. The independence assumption is not derived from equilibrium lattice dynamics; it defines a tractable reference class for which the hypotheses of the continuum-limit theorem can be verified. Displaced-structure ensembles and ionic-crystal simulations motivate random perturbations about mean positions, although the corresponding physical displacements are generally correlated across cells \cite{ZachariasGiustino2020,GangemiEtAl2022}. Let $N_c$ denote the number of charge profiles in each cell. For each $\bk\in\Z^3$, let
\begin{equation*}
    \mathrm{u}_{\bk}(\bx,\omega)
    =
    \{\bu_{\bk,\alpha}(\bx,\omega)\}_{\alpha=1}^{N_c}
\end{equation*}
denote the collection of random displacement fields in cell $\bk$. Assume that these collections are identically distributed and that $\{\mathrm{u}_{\bk}\}_{\bk\in\Z^3}$ is independent. Displacements of different charges within the same cell may be correlated. We assume that $\bu_{\bk,\alpha}(\cdot,\omega)$ is continuous in $\bx$. Let $q_\alpha$ be the charge of profile $\alpha$ and let $\ba_\alpha:\overline\Omega\to Q$ be its continuous reference-position field, with
\begin{equation}
    \sum_{\alpha=1}^{N_c}q_\alpha=0.
    \label{eq:random_cell_charge_neutrality}
\end{equation}
Let $\eta_\alpha(\bx,\cdot)\in C_c^5(\R^3)$, with $\bx\mapsto\eta_\alpha(\bx,\cdot)$ continuous in $L^2(\R^3)$ and
\begin{equation*}
    \sup_{\bx\in\overline\Omega}
    \|\eta_\alpha(\bx,\cdot)\|_{L^\infty(\R^3)}<\infty.
\end{equation*}
Assume
\begin{equation*}
    \int_{\R^3}\eta_\alpha(\bx,\bs)\dd\bs=1,
    \qquad
    \int_{\R^3}\bs\,\eta_\alpha(\bx,\bs)\dd\bs=\bm0.
\end{equation*}
For some $\delta>0$, require
\begin{equation*}
    \supp\!\left[
    \eta_\alpha\bigl(
    \bx,\cdot-\ba_\alpha(\bx)-\bu_{\bk,\alpha}(\bx,\omega)
    \bigr)
    \right]
    \subset
    \left\{
    \bs\in Q:
    \operatorname{dist}(\bs,\partial Q)\geq\delta
    \right\}
\end{equation*}
for every $\alpha$, $\bx\in\overline\Omega$, $\bk\in\Z^3$, and $\omega\in\Xi$. Define
\begin{equation}
    \bar\rho(\bx,\bs,\omega)
    =
    \sum_{\alpha=1}^{N_c}
    q_\alpha
    \eta_\alpha
    \bigl(
    \bx,
    \bs-\ba_\alpha(\bx)-\bu_{\bm0,\alpha}(\bx,\omega)
    \bigr).
    \label{eq:random_reference_profile}
\end{equation}
On the corresponding product probability space, $T_{\bk}$ shifts the cell variables. The action is measure preserving and ergodic. The support and normalization of $\eta_\alpha$ give
\begin{equation}
    \int_Q
    \bar\rho(\bx,\bs,\omega)\dd\bs
    =
    \sum_{\alpha=1}^{N_c}q_\alpha
    =
    0
    \label{eq:random_profile_cell_neutrality}
\end{equation}
for every realization. Similarly, using the zero first moment of $\eta_\alpha$,
\begin{equation}
    \int_Q
    \bs\,\bar\rho(\bx,\bs,\omega)\dd\bs
    =
    \sum_{\alpha=1}^{N_c}
    q_\alpha
    \bigl(
    \ba_\alpha(\bx)+\bu_{\bm0,\alpha}(\bx,\omega)
    \bigr).
    \label{eq:random_cell_polarization}
\end{equation}
The effective polarization is therefore
\begin{equation}
    \bp(\bx)
    =
    \sum_{\alpha=1}^{N_c}
    q_\alpha
    \left(
    \ba_\alpha(\bx)+
    \E[\bu_{\bm0,\alpha}(\bx,\cdot)]
    \right).
    \label{eq:random_effective_polarization}
\end{equation}
Thus the dipole moments of individual cells fluctuate, while their
large-volume average is deterministic. In particular, for centered
displacements,
\begin{equation}
    \bp(\bx)
    =
    \sum_{\alpha=1}^{N_c}
    q_\alpha\ba_\alpha(\bx),
    \label{eq:centered_random_polarization}
\end{equation}
which is the polarization of the unperturbed periodic configuration.

Exact neutrality and the stated continuity conditions give \textbf{A1}. To verify \textbf{A2}, we separate the charge density into its mean profile and independent centered cell fluctuations. Neutrality makes the squared $L^2(Q\times\Xi)$ norms of the centered cell potentials summable, while independence yields convergence of their sum. For \textbf{A3}, the centered contributions to the local field correction and the translated finite-volume field converge strongly to zero; the mean profile satisfies the periodic limits. The verification is given in \cref{app:random_displacement_fields}.

The example also illustrates the different information retained by the two parts of the limiting energy. For centered displacements, the effective polarization, and hence the macroscopic nonlocal energy, are unchanged from the unperturbed periodic configuration. The local energy density, however, consists of
\begin{equation}
    e_{\mathrm{loc}}(\bx)
    =
    \frac12
    \E\!\left[
    \int_Q
    \bar\rho(\bx,\bs,\cdot)
    \bar h(\bx,\bs,\cdot)
    \dd\bs
    \right]
    +
    \frac16|\bp(\bx)|^2.
    \label{eq:random_local_energy}
\end{equation}
The ensemble-averaged term retains the displacement statistics, whereas the
polarization term is unchanged by centered displacements. Random microscopic
fluctuations can thus change the local electrostatic energy without changing
the macroscopic polarization.

The periodic setting considered in
\cite{MarshallDayal2014,SenWangBreitzmanDayal2026} is recovered when the
displacements vanish. Here, different cells in the same realization may
instead have different charge profiles and dipole moments, while ergodicity
produces a deterministic continuum polarization. The separation between
the deterministic continuum polarization and the fluctuating cell dipoles
is related to the distinction between macroscopic dipole fields and fine-scale dipole
oscillations in \cite{JamesMuller1994,MullerSchlomerkemper2002,SchlomerkemperSchmidt2009}; in the present setting, however,
the fluctuating dipoles arise from an underlying random charge density.

\subsection{Numerical illustration}
\label{sec:numerical_illustration}

We consider the random-displacement model of \cref{sec:random_example}
in the unit cube $\Omega$ centered at the origin, partitioned into $M^3$ material-point regions of side $\veps=1/M$. Each region contains $R^3$ microscopic cells of side $\ell=\veps/R$, where $M$ and $R$ are odd. Each microscopic cell contains two normalized radial charge profiles with charges $q_\pm=\pm1$ and reference positions
\begin{equation*}
    \ba_\pm(\bx)=\pm \frac{d(x_1)}{2}\be_1,
    \qquad
    d(x_1)=0.25+0.08\sin(2\pi x_1).
\end{equation*}
The profile shape is proportional to $(1-|\bs|^2/a^2)_+^6$, with
$a=0.01$ in cell units.

Let the random variables $S_{\bk,\pm}$, $\bk\in\Z^3$, be mutually independent, each taking the values $\pm1$ with probability $1/2$. The displacement field is
\begin{equation*}
    \bu_{\bk,\pm}(\bx,\omega)
    =
    \sigma S_{\bk,\pm}(\omega)\be_1 .
\end{equation*}
Thus, for a fixed realization $\omega$, the pair $(S_{\bk,+}(\omega),S_{\bk,-}(\omega))$ specifies the displacements of the two charges in cell $\bk$, and the same pair is used at every material point $\bx$; see \cref{fig:random_model_schematic}. We compare the periodic case $\sigma=0$ with $\sigma=0.015$ and $\sigma=0.03$. These values keep the displaced charge profiles disjoint and contained in their cells. Since the displacements are centered,
\begin{equation*}
    \bp(\bx)=d(x_1)\be_1
\end{equation*}
for all three cases.

\begin{figure}[htbp]
\centering
\includegraphics[width=0.95\textwidth]{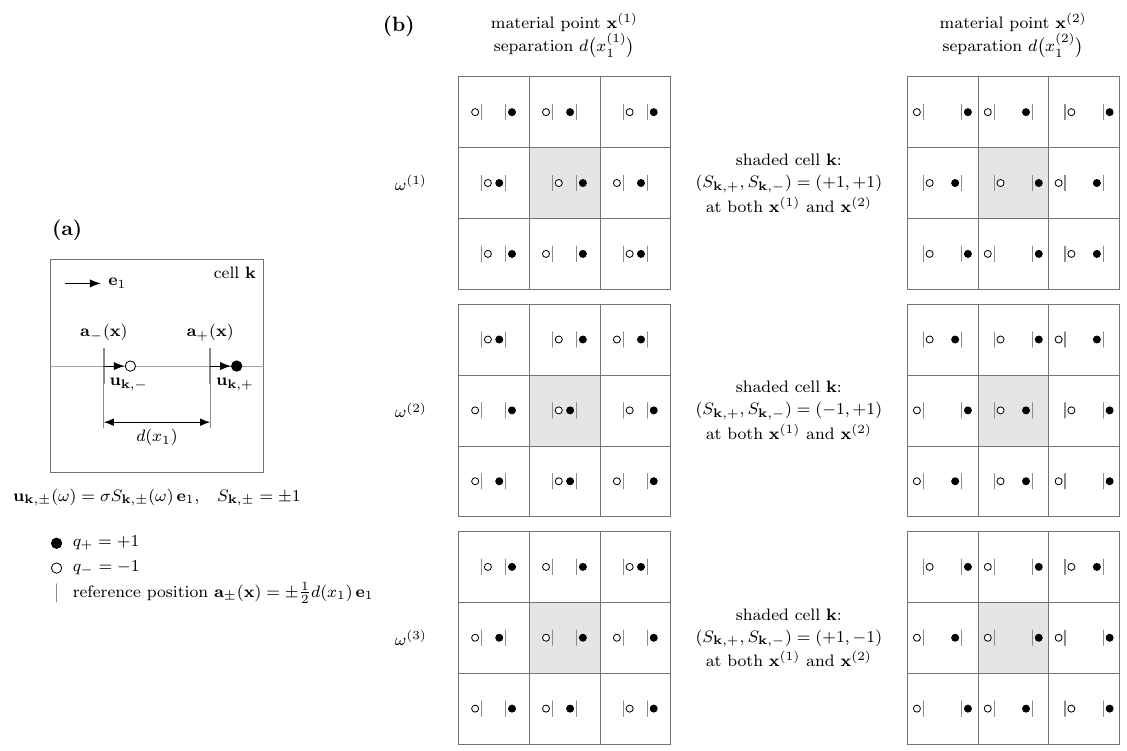}
\vspace{-10pt}
\caption{Schematic of the centered random-displacement model (not to scale). (a) A microscopic cell $\bk$ at the material point $\bx$ contains two charge profiles with charges $q_+=+1$ and $q_-=-1$. Their reference positions are $\ba_\pm(\bx)=\pm\tfrac12 d(x_1)\be_1$, and each profile is displaced by $\bu_{\bk,\pm}(\omega)=\sigma S_{\bk,\pm}(\omega)\be_1$, where the signs $S_{\bk,\pm}=\pm1$ are independent and equally likely. (b) Portions of the microscopic lattice at two material points $\bx^{(1)}$ and $\bx^{(2)}$ for three realizations $\omega^{(1)},\omega^{(2)},\omega^{(3)}$. For a fixed realization, the same signs $(S_{\bk,+},S_{\bk,-})$ are used in cell $\bk$ at both material points, as indicated for the shaded cell, while the reference separation $d(x_1)$ differs between $\bx^{(1)}$ and $\bx^{(2)}$. The separations and $\sigma$ are exaggerated for visibility.}
\label{fig:random_model_schematic}
\end{figure}

The finite material-point polarization satisfies
\begin{equation*}
    \bp_R(\bx,\omega)-\bp(\bx)
    =\frac{\sigma}{R^3}
    \sum_{\bk\in I_R}
    (S_{\bk,+}-S_{\bk,-})\be_1.
\end{equation*}
For fixed $R$ and $\omega$, this difference does not depend on $\bx$. In \cref{fig:numerical_polarization}(a), the sampled values of $\bp_R\cdot\be_1$ lie closer to $d(x_1)$ at $R=11$ than at $R=3$. The root-mean-square differences in panel (b) decrease with increasing $R$ for both values of $M$ and both displacement magnitudes. These comparisons illustrate convergence toward the polarization used to evaluate the effective energy below.

\begin{figure}[htbp]
\centering
\includegraphics[width=0.85\textwidth]{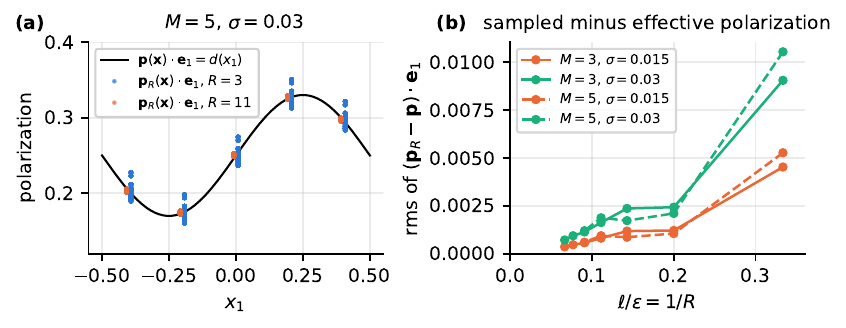}
\vspace{-15pt}
\caption{Material-point polarization in the centered random-displacement model. (a) Values of $\bp_R(\bx,\omega)\cdot\be_1$ over all material points and 16 realizations for $M=5$, $R=3,11$, and $\sigma=0.03$, compared with $\bp(\bx)\cdot\be_1=d(x_1)$. The point sets have small horizontal offsets for visibility. (b) Root-mean-square of $(\bp_R-\bp)\cdot\be_1$ over material points and realizations for $M=3,5$, using 32 and 16 realizations, respectively, and $\sigma=0.015,0.03$.}
\label{fig:numerical_polarization}
\end{figure}

We compute the finite-scale energies by direct Coulomb pair sums between distinct profiles, using the scaling in \cref{eq:scaled_charge_density}, with analytic self-energies included in the local part. Radial symmetry and disjoint supports make each pair interaction equal to that between point charges at the profile centers. We vary both $R$ and $M$ and compare the finite-scale energies with their independently evaluated limits as $R\to\infty$ at fixed $M$:
\begin{equation*}
\begin{aligned}
    \cE_{\veps}^{\mathrm{loc}}(\sigma)
    &:=\veps^3\sum_{\bx\in\Omega_\veps}e_{\mathrm{loc}}(\bx;\sigma),\\
    \cE_{\veps}^{\mathrm{nloc}}
    &:=\frac12\int_{\R^3}|\nabla\phi_\veps|^2\dd\bx
    -\frac{\veps^3}{6}\sum_{\bx\in\Omega_\veps}|\bp(\bx)|^2,
\end{aligned}
\end{equation*}
where $e_{\mathrm{loc}}(\bx;\sigma)$ is given by \cref{eq:random_local_energy} for displacement magnitude $\sigma$, and $\phi_\veps$ is the potential generated by $\bp_\veps$ in \cref{eq:piecewise_constant_polarization,eq:piecewise_polarization_problem}. Their sum converges to the effective total energy in \cref{eq:main_energy_limit} as $M\to\infty$, by \cref{sec:local_limit,sec:nonlocal_limit}. For this independent-cell law, \cref{eq:random_cell_energy_expectation} gives the microscopic expectation from the periodic mean profile and the Coulomb energy of the centered cell fluctuation. The latter is averaged over the four equally likely displacement states.

\begin{figure}[!t]
\centering
\includegraphics[width=0.9\textwidth]{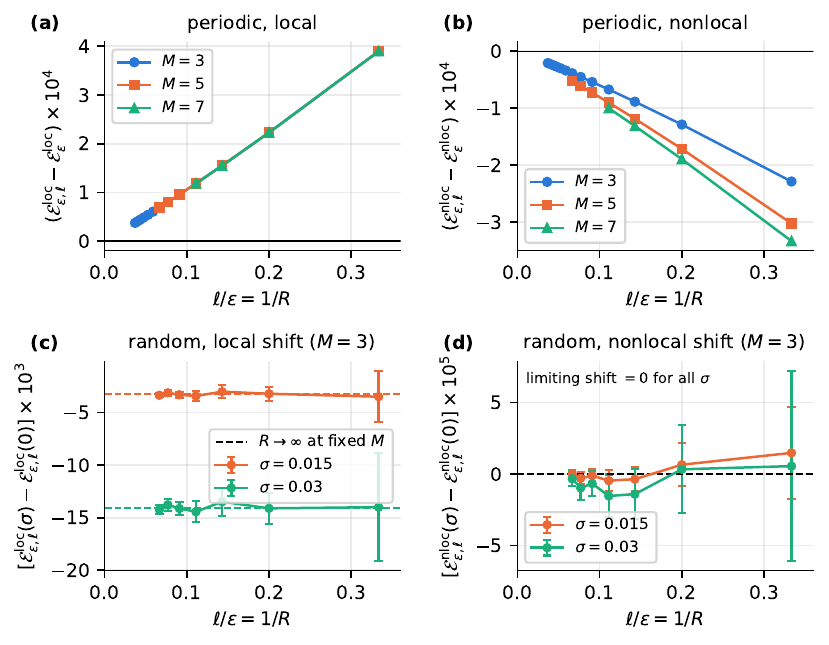}
\vspace{-15pt}
\caption{Finite-scale energies compared with independently evaluated limits as $R\to\infty$ at fixed $M$. (a,b) Periodic local and nonlocal energies minus their respective limiting values, for $M=3,5,7$. (c,d) Sample mean local and nonlocal energy shifts relative to the periodic configuration at the same $R$, for $M=3$ and $\sigma=0.015,0.03$. Dashed lines show the limiting local-energy shifts in (c) and the zero limiting nonlocal shift in (d). Error bars are twice the standard error of the mean over 32 realizations.}
\label{fig:numerical_energies}
\end{figure}

In \cref{fig:numerical_energies}(a,b), the periodic local and nonlocal energies approach their respective limits as $R$ increases for $M=3,5,7$. Panel (c) shows that random displacements reduce the average local energy relative to the periodic configuration. The reduction is larger for $\sigma=0.03$ than for $\sigma=0.015$ and agrees closely with the values calculated from \cref{eq:random_local_energy}, shown by the dashed lines. In panel (d), the average nonlocal energy changes much less. Since centered displacements leave $\bp$ unchanged, the random and periodic nonlocal energies approach the same value. Thus, the change in total continuum energy is entirely due to the change in local energy.

The continuum model reduces computational work by avoiding the Coulomb pair sum over the full microscopic specimen. For $M=3$ and $R=27$, direct summation involves $1.06\times10^6$ charge profiles. Instead, the limiting energy is evaluated using a periodic-cell calculation and a macroscopic field calculation; at fixed $M$ and numerical accuracy, this computation does not grow with $R$. For this periodic case, the relative differences between the finite-scale energies and their $R\to\infty$ values at $M=3$ are $2.1 \times 10^{-4}\%$ for the local energy and $4.09\%$ for the nonlocal energy. The total energy is dominated by the local contribution, including the profile self-energies, so the larger relative difference in the much smaller nonlocal contribution has little effect on the total energy. Moreover, the positive local difference and negative nonlocal difference partially cancel, leaving a total relative difference of only $9.4 \times 10^{-5}\%$.

\section{Conclusions}
\label{sec:conclusions}

The continuum limit shows that ergodic averaging does not remove all dependence of the electrostatic energy on microscopic randomness. For interactions between distinct material points, the fluctuating cell dipoles reduce to the deterministic effective polarization. Interactions within a material point instead retain the ensemble-averaged microscopic Coulomb energy. The cell-depolarization contribution associated with the cubic decomposition belongs only to the local--nonlocal partition: its equal and opposite terms cancel in the total energy. The resulting total continuum energy combines the ensemble-averaged microscopic cell energy with the macroscopic field energy.

The periodic specialization recovers the energy obtained in \cite{MarshallDayal2014}, while the random-displacement model verifies \textbf{A1}--\textbf{A3} for a nonperiodic class. Centered cellwise displacements in the latter model can change the local and total energies without changing either the effective polarization or the nonlocal energy. A continuum description based only on effective polarization therefore misses the energetic effect of those microscopic fluctuations.

The theorem applies to stationary and ergodic charge fields satisfying \textbf{A1}--\textbf{A3}, but the explicit verification for the random-displacement model uses independence across cells. Correlated displacement fields require a separate verification of the finite-volume assumptions. Temporal charge dynamics and material-point geometries other than the fixed cube used here are also outside the present analysis.

\section*{Acknowledgments}
PKJ acknowledges support from the National Science Foundation through the Engineering Research Initiation (ERI) program under Award No.~2502279 and from the South Dakota Board of Regents Competitive Research Grant (SDBOR CRG) program. KD acknowledges support from the Army Research Office through MURI Awards W911NF-24-2-0184 and W911NF-25-2-0164. Part of this work was carried out during PKJ's doctoral studies at Carnegie Mellon University and appeared in his Ph.D. thesis \cite{Jha2016}. Any opinions, findings, and conclusions or recommendations expressed in this material are those of the authors and do not necessarily reflect the views of the funding agencies.


\appendix

\section{Ergodicity of the lattice action: Proofs}\label{sec:ergodicityProof}

\begin{proof}[Proof of \cref{thm:ergodic}]
We derive the result from the Birkhoff ergodic theorem for
$\R^3$-parameterized dynamical systems \cite[Theorem~7.2]{JikovKozlovOleinik1994}.

Introduce the enlarged probability space
\begin{equation*}
    \widetilde\Xi:=\Xi\times Q,
    \qquad
    \widetilde\Pp
    :=
    \Pp\otimes\left.\mathcal L^3\right|_Q.
\end{equation*}
Since $|Q|=1$, $\widetilde\Pp$ is a probability measure. For every $\by\in\R^3$ and $\bs\in Q$, define
$\bk_{\by}(\bs)\in\Z^3$ and $\tau_{\by}(\bs)\in Q$ by the unique
decomposition
\begin{equation}
    \bs+\by
    =
    \bk_{\by}(\bs)+\tau_{\by}(\bs).
    \label{eq:lattice_fractional_decomposition}
\end{equation}
The map $\bk_{\by}(\bs)$ gives the lattice cell containing $\bs+\by$,
while $\tau_{\by}(\bs)$ gives its coordinate within the reference cell.
Define
\begin{equation}
    \widetilde T_{\by}(\omega,\bs)
    :=
    \left(
    T_{\bk_{\by}(\bs)}\omega,
    \tau_{\by}(\bs)
    \right).
    \label{eq:extended_lattice_action}
\end{equation}

We first verify the group property. Applying
\cref{eq:lattice_fractional_decomposition} successively with $\bz$ and
$\by$ gives
\begin{equation*}
    \bs+\bz
    =
    \bk_{\bz}(\bs)+\tau_{\bz}(\bs)
\end{equation*}
and
\begin{equation*}
    \tau_{\bz}(\bs)+\by
    =
    \bk_{\by}\bigl(\tau_{\bz}(\bs)\bigr)
    +
    \tau_{\by}\bigl(\tau_{\bz}(\bs)\bigr).
\end{equation*}
Adding these identities yields
\begin{equation*}
\begin{aligned}
    \bs+\by+\bz
    &=
    \bk_{\bz}(\bs)
    +
    \bk_{\by}\bigl(\tau_{\bz}(\bs)\bigr)
    +
    \tau_{\by}\bigl(\tau_{\bz}(\bs)\bigr).
\end{aligned}
\end{equation*}
By uniqueness of the decomposition into a lattice coordinate and a
reference-cell coordinate,
\begin{equation}
\begin{aligned}
    \bk_{\by+\bz}(\bs)
    &=
    \bk_{\bz}(\bs)
    +
    \bk_{\by}\bigl(\tau_{\bz}(\bs)\bigr),\\
    \tau_{\by+\bz}(\bs)
    &=
    \tau_{\by}\bigl(\tau_{\bz}(\bs)\bigr).
\end{aligned}
\label{eq:extended_action_composition}
\end{equation}
Using \cref{eq:action_properties,eq:extended_action_composition}, we obtain
\begin{equation*}
\begin{aligned}
    \widetilde T_{\by}
    \left(
    \widetilde T_{\bz}(\omega,\bs)
    \right)
    &=
    \widetilde T_{\by}
    \left(
    T_{\bk_{\bz}(\bs)}\omega,
    \tau_{\bz}(\bs)
    \right)\\
    &=
    \left(
    T_{\bk_{\by}(\tau_{\bz}(\bs))}
    T_{\bk_{\bz}(\bs)}\omega,
    \tau_{\by}(\tau_{\bz}(\bs))
    \right)\\
    &=
    \left(
    T_{\bk_{\by+\bz}(\bs)}\omega,
    \tau_{\by+\bz}(\bs)
    \right)\\
    &=
    \widetilde T_{\by+\bz}(\omega,\bs).
\end{aligned}
\end{equation*}
Also,
\begin{equation*}
    \widetilde T_{\bm0}
    =
    \operatorname{Id}.
\end{equation*}
Therefore, $\{\widetilde T_{\by}\}_{\by\in\R^3}$ is an
$\R^3$-parameterized action.

We next verify that the action preserves $\widetilde\Pp$. For fixed
$\by\in\R^3$ and $\bk\in\Z^3$, define
\begin{equation*}
    E_{\bk}(\by)
    :=
    \left\{
    \bs\in Q:
    \bk_{\by}(\bs)=\bk
    \right\}.
\end{equation*}
The sets $E_{\bk}(\by)$ form a disjoint partition of $Q$. On
$E_{\bk}(\by)$,
\begin{equation*}
    \tau_{\by}(\bs)
    =
    \bs+\by-\bk.
\end{equation*}
Moreover, the sets
\begin{equation*}
    E_{\bk}(\by)+\by-\bk,
    \qquad
    \bk\in\Z^3,
\end{equation*}
also form a disjoint partition of $Q$.

Let $\Phi\in L^1(\widetilde\Xi)$. Using the measure-preserving property of
$T_{\bk}$ and the change of variables
$\boldsymbol t=\bs+\by-\bk$, we obtain
\begin{equation*}
\begin{aligned}
    &\int_Q\int_\Xi
    \Phi\bigl(\widetilde T_{\by}(\omega,\bs)\bigr)
    \dd\Pp(\omega)\dd\bs\\
    &\quad=
    \sum_{\bk\in\Z^3}
    \int_{E_{\bk}(\by)}
    \int_\Xi
    \Phi\left(
    T_{\bk}\omega,
    \bs+\by-\bk
    \right)
    \dd\Pp(\omega)\dd\bs\\
    &\quad=
    \sum_{\bk\in\Z^3}
    \int_{E_{\bk}(\by)+\by-\bk}
    \int_\Xi
    \Phi(\omega,\boldsymbol t)
    \dd\Pp(\omega)\dd\boldsymbol t\\
    &\quad=
    \int_Q\int_\Xi
    \Phi(\omega,\boldsymbol t)
    \dd\Pp(\omega)\dd\boldsymbol t.
\end{aligned}
\end{equation*}
Thus, $\widetilde T_{\by}$ preserves $\widetilde\Pp$.

Finally, we verify ergodicity. Let
$\Phi\in L^1(\widetilde\Xi)$ satisfy
\begin{equation}
    \Phi\bigl(\widetilde T_{\by}(\omega,\bs)\bigr)
    =
    \Phi(\omega,\bs)
    \label{eq:extended_invariant_function}
\end{equation}
for every $\by\in\R^3$ and almost every
$(\omega,\bs)\in\widetilde\Xi$.

For $\by=\bk\in\Z^3$, we have
\begin{equation*}
    \bk_{\bk}(\bs)=\bk,
    \qquad
    \tau_{\bk}(\bs)=\bs,
\end{equation*}
and therefore
\begin{equation*}
    \widetilde T_{\bk}(\omega,\bs)
    =
    (T_{\bk}\omega,\bs).
\end{equation*}
It follows from \cref{eq:extended_invariant_function} that
\begin{equation*}
    \Phi(T_{\bk}\omega,\bs)
    =
    \Phi(\omega,\bs)
\end{equation*}
for every $\bk\in\Z^3$ and almost every
$(\omega,\bs)\in\Xi\times Q$. Since $\Z^3$ is countable, Fubini's theorem
implies that, for almost every fixed $\bs\in Q$,
\begin{equation*}
    \Phi(T_{\bk}\omega,\bs)
    =
    \Phi(\omega,\bs)
\end{equation*}
for every $\bk\in\Z^3$ and almost every $\omega\in\Xi$.

Fix such a value of $\bs$ and define
\begin{equation*}
    \Phi_{\bs}(\omega):=\Phi(\omega,\bs).
\end{equation*}
Then $\Phi_{\bs}$ is invariant under the original lattice action:
\begin{equation*}
    \Phi_{\bs}(T_{\bk}\omega)
    =
    \Phi_{\bs}(\omega)
\end{equation*}
for every $\bk\in\Z^3$ and almost every $\omega\in\Xi$. To see why
ergodicity implies that $\Phi_{\bs}$ is constant, consider, for any
$a\in\R$, the measurable set
\begin{equation*}
    A_a(\bs)
    :=
    \left\{
    \omega\in\Xi:
    \Phi_{\bs}(\omega)>a
    \right\}.
\end{equation*}
The invariance of $\Phi_{\bs}$ gives
\begin{equation*}
    T_{\bk}^{-1}A_a(\bs)
    =
    A_a(\bs)
\end{equation*}
for every $\bk\in\Z^3$. By ergodicity of the original lattice action,
\begin{equation*}
    \Pp\bigl(A_a(\bs)\bigr)
    \in\{0,1\}.
\end{equation*}
Since every level set of $\Phi_{\bs}$ has probability either zero or one,
$\Phi_{\bs}$ must be constant almost surely. Hence, for almost every
$\bs\in Q$, there exists a number $c(\bs)$ such that
\begin{equation}
    \Phi(\omega,\bs)
    =
    c(\bs)
    \label{eq:extended_invariant_cell_function}
\end{equation}
for almost every $\omega\in\Xi$.

Substituting \cref{eq:extended_invariant_cell_function} into
\cref{eq:extended_invariant_function} gives
\begin{equation}
    c\bigl(\tau_{\by}(\bs)\bigr)
    =
    c(\bs)
    \label{eq:cell_translation_invariance}
\end{equation}
for every $\by\in\R^3$ and almost every $\bs\in Q$. By Fubini's theorem,
\cref{eq:cell_translation_invariance} holds for almost every
$(\bs,\by)\in Q\times Q$. Hence, for almost every $\bs\in Q$,
\begin{equation*}
    c(\bs)
    =
    \int_Q c\bigl(\tau_{\by}(\bs)\bigr)\dd\by
    =
    \int_Q c(\boldsymbol t)\dd\boldsymbol t.
\end{equation*}
In the last equality, we used that
$\by\mapsto\tau_{\by}(\bs)$ is translation modulo $\Z^3$ and therefore
preserves Lebesgue measure on $Q$. The right-hand side is independent of
$\bs$, so $c$ is constant almost everywhere. Consequently, every invariant
function $\Phi$ is constant almost everywhere, and the extended action is
ergodic.

Define
\begin{equation*}
    \widetilde f(\omega,\bs):=f(\omega).
\end{equation*}
Theorem~7.2 of \cite{JikovKozlovOleinik1994}, applied to
$\widetilde f$ and $|\widetilde f|$, gives, for almost every
$(\omega,\bs)\in\widetilde\Xi$,
\begin{equation}
    \frac{1}{R^3}
    \int_{Q_R}
    \widetilde f\bigl(
    \widetilde T_{\by}(\omega,\bs)
    \bigr)
    \dd\by
    \longrightarrow
    \E[f].
    \label{eq:extended_spatial_average}
\end{equation}

For fixed $\omega$, define the piecewise-constant realization
\begin{equation*}
    F_\omega(\bk+\boldsymbol t)
    :=
    f(T_{\bk}\omega),
    \qquad
    \bk\in\Z^3,
    \quad
    \boldsymbol t\in Q.
\end{equation*}
For $R=2n+1$, the decomposition
\begin{equation*}
    Q_R
    =
    \biguplus_{\bk\in I_R}(\bk+Q)
\end{equation*}
gives
\begin{equation}
    \frac{1}{R^3}
    \sum_{\bk\in I_R}
    f(T_{\bk}\omega)
    =
    \frac{1}{R^3}
    \int_{Q_R}
    F_\omega(\by)\dd\by.
    \label{eq:lattice_average_integral}
\end{equation}
On the other hand,
\begin{equation*}
    \widetilde f\bigl(
    \widetilde T_{\by}(\omega,\bs)
    \bigr)
    =
    F_\omega(\by+\bs),
\end{equation*}
and therefore
\begin{equation*}
    \frac{1}{R^3}
    \int_{Q_R}
    \widetilde f\bigl(
    \widetilde T_{\by}(\omega,\bs)
    \bigr)
    \dd\by
    =
    \frac{1}{R^3}
    \int_{Q_R+\bs}
    F_\omega(\by)\dd\by.
\end{equation*}

Since $\bs\in Q$, the sets $Q_R$ and $Q_R+\bs$ differ only in a boundary
layer of fixed thickness. The volume of this layer is of order $R^2$,
whereas $|Q_R|=R^3$. Applying
\cref{eq:extended_spatial_average} to $|f|$ shows that the contribution
of this boundary layer, divided by $R^3$, converges to zero. Consequently,
the averages over $Q_R$ and $Q_R+\bs$ have the same limit.

By Fubini's theorem, for almost every $\omega\in\Xi$, there exists
$\bs\in Q$ for which the preceding convergence holds. Combining it with
\cref{eq:extended_spatial_average,eq:lattice_average_integral} gives
\begin{equation*}
    \frac{1}{R^3}
    \sum_{\bk\in I_R}
    f(T_{\bk}\omega)
    \longrightarrow
    \E[f].
\end{equation*}

For the $L^2$ convergence, let
$\boldsymbol e_1,\boldsymbol e_2,\boldsymbol e_3$ be the standard basis
of $\Z^3$ and define
\begin{equation*}
    U_jf:=f\circ T_{\boldsymbol e_j},
    \qquad
    j=1,2,3.
\end{equation*}
For $R=2n+1$, the lattice average can be written as
\begin{equation*}
    \frac{1}{R^3}
    \sum_{\bk\in I_R}
    f\circ T_{\bk}
    =
    \prod_{j=1}^3
    \left(
    \frac{1}{2n+1}
    \sum_{m=-n}^{n}U_j^m
    \right)f.
\end{equation*}
The mean ergodic theorem shows that each factor converges in
$L^2(\Xi)$ to the projection onto the functions invariant under
$T_{\boldsymbol e_j}$. Since the three transformations commute, their
product converges to the projection onto the functions invariant under
every $T_{\bk}$. Ergodicity implies that such functions are constant.
The limiting constant is therefore $\E[f]$.
\end{proof}

\begin{proof}[Proof of \cref{cor:weighted_ergodic}]
For fixed $\omega$, define the piecewise-constant realization
\begin{equation*}
    f_\omega(\bk+\bs)
    :=
    f(T_{\bk}\omega),
    \qquad
    \bk\in\Z^3,
    \quad
    \bs\in Q.
\end{equation*}
The mean-value statement used in the proof of
\cref{thm:ergodic} gives
\begin{equation*}
    f_\omega(R\,\cdot)
    \rightharpoonup
    \E[f]
    \qquad
    \text{in }L^2(Q;\R^m)
\end{equation*}
for almost every $\omega\in\Xi$. Hence,
\begin{equation}
    \int_Q
    a(\boldsymbol\eta)
    f_\omega(R\boldsymbol\eta)
    \dd\boldsymbol\eta
    \longrightarrow
    \left(\int_Qa(\boldsymbol\eta)\dd\boldsymbol\eta\right)
    \E[f].
    \label{eq:weighted_realization_limit}
\end{equation}

Using
\begin{equation*}
    Q_R
    =
    \biguplus_{\bk\in I_R}(\bk+Q)
\end{equation*}
and the change of variables
$\boldsymbol\eta=(\bk+\bs)/R$, we obtain
\begin{equation*}
\begin{aligned}
    \int_Q
    a(\boldsymbol\eta)
    f_\omega(R\boldsymbol\eta)
    \dd\boldsymbol\eta
    &=
    \frac{1}{R^3}
    \sum_{\bk\in I_R}
    \int_Q
    a\left(\frac{\bk+\bs}{R}\right)
    f(T_{\bk}\omega)
    \dd\bs.
\end{aligned}
\end{equation*}
Therefore,
\begin{equation*}
\begin{aligned}
    &
    \left|
    \int_Q
    a(\boldsymbol\eta)
    f_\omega(R\boldsymbol\eta)
    \dd\boldsymbol\eta
    -
    \frac{1}{R^3}
    \sum_{\bk\in I_R}
    a\left(\frac{\bk}{R}\right)
    f(T_{\bk}\omega)
    \right|\\
    &\qquad\leq
    \omega_a\left(\frac{\sqrt{3}}{2R}\right)
    \frac{1}{R^3}
    \sum_{\bk\in I_R}
    |f(T_{\bk}\omega)|,
\end{aligned}
\end{equation*}
where
\begin{equation*}
    \omega_a(\delta)
    :=
    \sup_{\substack{
    \boldsymbol\xi,\boldsymbol\eta\in\overline Q\\
    |\boldsymbol\xi-\boldsymbol\eta|\leq\delta}}
    |a(\boldsymbol\xi)-a(\boldsymbol\eta)|
\end{equation*}
is the modulus of continuity of $a$. Since $a$ is uniformly continuous,
\begin{equation*}
    \omega_a\left(\frac{\sqrt{3}}{2R}\right)
    \longrightarrow0.
\end{equation*}
Moreover, \cref{thm:ergodic} applied to $|f|$ gives
\begin{equation*}
    \frac{1}{R^3}
    \sum_{\bk\in I_R}
    |f(T_{\bk}\omega)|
    \longrightarrow
    \E[|f|].
\end{equation*}
Thus, the difference above converges to zero. Combining this result with
\cref{eq:weighted_realization_limit} proves the claim.
\end{proof}


\section{\texorpdfstring
{Regularity of microscopic potential $h$ and weak finite-volume limit of $\nabla (h_R - h)$}
{Regularity of microscopic potential h and weak finite-volume limit of the gradient of (hR - h)}}
\label{app:weak_potential}

This appendix proves \cref{prop:microscopic_potential_properties} and \cref{prop:weak_finite_volume_field}. 

\begin{proof}[Proof of \cref{prop:microscopic_potential_properties}]
The distributional equation follows from $h=G\ast\rho$ and
$-\Delta G=\delta_{\bm0}$. By \textbf{A1}--\textbf{A2},
\begin{equation*}
    h(\bx,\cdot,\omega)\in L^2_{\mathrm{loc}}(\R^3),
    \qquad
    \rho(\bx,\cdot,\omega)\in L^2_{\mathrm{loc}}(\R^3)
\end{equation*}
for every fixed $\bx$ and almost every $\omega$. Since
\begin{equation*}
    -\Delta_{\by}h(\bx,\cdot,\omega)
    =
    \rho(\bx,\cdot,\omega)
    \qquad
    \text{in }\mathcal D'(\R^3),
\end{equation*}
standard interior elliptic regularity implies
\begin{equation*}
    h(\bx,\cdot,\omega)\in H^2_{\mathrm{loc}}(\R^3),
\end{equation*}
and hence
\begin{equation*}
    h(\bx,\cdot,\omega)\in H^1_{\mathrm{loc}}(\R^3).
\end{equation*}

Choose fixed radii $0<r<R_\ast$ such that
\begin{equation*}
    Q\subset B_r(\bm0)
    \Subset
    B_{R_\ast}(\bm0).
\end{equation*}
The Caccioppoli inequality for the Poisson equation gives
\begin{equation}
    \int_{B_r(\bm0)}
    |\nabla_{\by}h(\bx,\by,\omega)|^2
    \dd\by
    \leq
    C
    \left[
    \int_{B_{R_\ast}(\bm0)}
    |h(\bx,\by,\omega)|^2
    \dd\by
    +
    \int_{B_{R_\ast}(\bm0)}
    |\rho(\bx,\by,\omega)|^2
    \dd\by
    \right],
    \label{eq:caccioppoli_h}
\end{equation}
where $C$ depends only on $r$ and $R_\ast$. Therefore,
\begin{equation}
    \int_Q
    |\nabla_{\by}h(\bx,\by,\omega)|^2
    \dd\by
    \leq
    C
    \left[
    \int_{B_{R_\ast}(\bm0)}
    |h(\bx,\by,\omega)|^2
    \dd\by
    +
    \int_{B_{R_\ast}(\bm0)}
    |\rho(\bx,\by,\omega)|^2
    \dd\by
    \right].
    \label{eq:local_gradient_bound}
\end{equation}

Since $B_{R_\ast}(\bm0)$ is bounded, it is contained in a finite union of
lattice cells,
\begin{equation*}
    B_{R_\ast}(\bm0)
    \subset
    \bigcup_{j=1}^{N}
    (\bk_j+Q),
    \qquad
    \bk_j\in\Z^3.
\end{equation*}
Using lattice stationarity,
\begin{equation*}
\begin{aligned}
    \E\!\left[
    \int_{B_{R_\ast}(\bm0)}
    |h(\bx,\by,\cdot)|^2
    \dd\by
    \right]
    &\leq
    \sum_{j=1}^{N}
    \E\!\left[
    \int_{\bk_j+Q}
    |h(\bx,\by,\cdot)|^2
    \dd\by
    \right]\\
    &=
    N
    \E\!\left[
    \int_Q
    |\bar h(\bx,\bs,\cdot)|^2
    \dd\bs
    \right]
    <\infty.
\end{aligned}
\end{equation*}
Similarly,
\begin{equation*}
    \E\!\left[
    \int_{B_{R_\ast}(\bm0)}
    |\rho(\bx,\by,\cdot)|^2
    \dd\by
    \right]
    \leq
    N
    \E\!\left[
    \int_Q
    |\bar\rho(\bx,\bs,\cdot)|^2
    \dd\bs
    \right]
    <\infty.
\end{equation*}
Taking expectations in \cref{eq:local_gradient_bound} therefore yields
\cref{eq:stationary_field_energy_bound}.
\end{proof}

\begin{proof}[Proof of \cref{prop:weak_finite_volume_field}]
Fix $\bx\in\overline\Omega$ and suppress $\bx$ in the notation.

Define
\begin{equation}
    \mu_R
    :=
    R\rho(R\,\cdot,\omega)\chi_Q.
    \label{eq:appendix_scaled_charge}
\end{equation}
For $\varphi\in C_c^\infty(\R^3)$, cell neutrality gives
\begin{equation}
\begin{aligned}
    \langle\mu_R,\varphi\rangle
    &=
    \frac{1}{R^2}
    \sum_{\bk\in I_R}
    \int_Q
    \bar\rho(\bs,T_{\bk}\omega)
    \left[
    \varphi\left(\frac{\bk+\bs}{R}\right)
    -
    \varphi\left(\frac{\bk}{R}\right)
    \right]
    \dd\bs.
\end{aligned}
\label{eq:appendix_scaled_charge_neutrality}
\end{equation}
Taylor's formula gives
\begin{equation}
\begin{aligned}
    \langle\mu_R,\varphi\rangle
    =
    \frac{1}{R^3}
    \sum_{\bk\in I_R}
    \nabla\varphi\left(\frac{\bk}{R}\right)
    \cdot
    \int_Q
    \bs\,\bar\rho(\bs,T_{\bk}\omega)
    \dd\bs
    +
    \mathcal R_R,
\end{aligned}
\label{eq:appendix_scaled_charge_dipole}
\end{equation}
where
\begin{equation*}
    |\mathcal R_R|
    \leq
    \frac{C_\varphi}{R}
    \frac{1}{R^3}
    \sum_{\bk\in I_R}
    \int_Q
    |\bar\rho(\bs,T_{\bk}\omega)|
    \dd\bs.
\end{equation*}
The last average is bounded almost surely by
\cref{thm:ergodic} and \textbf{A1}, so
$\mathcal R_R\to0$. Applying
\cref{cor:weighted_ergodic} to the cell dipole gives
\begin{equation}
    \mu_R
    \longrightarrow
    -\nabla\cdot(\bp\chi_Q)
    \qquad
    \text{in }\mathcal D'(\R^3).
    \label{eq:appendix_mu_limit}
\end{equation}

Set
\begin{equation}
    V_R(\boldsymbol\xi)
    :=
    \frac{h_R(R\boldsymbol\xi,\omega)}{R}.
    \label{eq:appendix_VR}
\end{equation}
A change of variables and the homogeneity of $G$ give
\begin{equation}
    V_R
    =
    G\ast\mu_R.
    \label{eq:appendix_VR_convolution}
\end{equation}
Since the distributions $\mu_R$ are supported in the fixed compact set
$\overline Q$, \cref{eq:appendix_mu_limit} implies
\begin{equation}
    V_R
    \longrightarrow
    G\ast[-\nabla\cdot(\bp\chi_Q)]
    =
    H_Q[\bp]
    \qquad
    \text{in }\mathcal D'(\R^3).
    \label{eq:appendix_VR_distributional}
\end{equation}

Let $\bar\bP$ be the cell polarization field from
\cref{lem:cell_polarization_field}, extended to $\R^3$ by
\cref{eq:cell_polarization_extension}. Since $Q_R$ is a union of
complete cells,
\begin{equation}
    -\Delta_{\boldsymbol\xi}V_R
    =
    -\nabla_{\boldsymbol\xi}\cdot
    \left[
    \bP(R\boldsymbol\xi,\omega)\chi_Q(\boldsymbol\xi)
    \right]
    \qquad
    \text{in }\mathcal D'(\R^3).
    \label{eq:appendix_VR_divergence}
\end{equation}
The $\dot H^1$ estimate for
$-\Delta v=-\nabla\cdot F$ gives
\begin{equation}
    \|\nabla_{\boldsymbol\xi}V_R\|_{L^2(\R^3)}
    \leq
    \|\bP(R\,\cdot,\omega)\|_{L^2(Q)}.
    \label{eq:appendix_VR_gradient_bound}
\end{equation}
By \cref{eq:cell_polarization_estimate,thm:ergodic}, the right-hand side
is bounded almost surely. Hence
$\{\nabla_{\boldsymbol\xi}V_R\}$ is bounded in
$L^2(\R^3;\R^3)$.

Every subsequence therefore has a further subsequence converging weakly
in $L^2(\R^3;\R^3)$. By
\cref{eq:appendix_VR_distributional}, the distributional limit of the
gradients is $\nabla H_Q[\bp]$. The weak limit is therefore unique, and
\begin{equation}
    \nabla_{\boldsymbol\xi}
    \left[
    \frac{h_R(R\boldsymbol\xi,\omega)}{R}
    \right]
    \rightharpoonup
    \nabla H_Q[\bp](\boldsymbol\xi)
    \qquad
    \text{weakly in }L^2(\R^3;\R^3).
    \label{eq:appendix_weak_VR_field}
\end{equation}
This proves \cref{eq:weak_truncated_finite_volume_field}.

Next set
\begin{equation}
    W_R(\boldsymbol\xi)
    :=
    \frac{h(R\boldsymbol\xi,\omega)}{R}.
    \label{eq:appendix_WR}
\end{equation}
By lattice stationarity, \textbf{A2}, and
\cref{thm:ergodic},
\begin{equation}
    \|W_R\|_{L^2(Q)}^2
    =
    \frac{1}{R^5}\int_{Q_R}|h(\by,\omega)|^2\dd\by
    =
    \frac{1}{R^2}
    \left[
    \frac{1}{R^3}\int_{Q_R}|h(\by,\omega)|^2\dd\by
    \right]
    \longrightarrow0.
\label{eq:appendix_WR_zero}
\end{equation}
Moreover,
\begin{equation}
    \|\nabla_{\boldsymbol\xi}W_R\|_{L^2(Q)}^2
    =
    \frac{1}{R^3}
    \int_{Q_R}
    |\nabla_{\by}h(\by,\omega)|^2
    \dd\by.
    \label{eq:appendix_WR_gradient}
\end{equation}
By \cref{eq:stationary_field_energy_bound,thm:ergodic}, the right-hand
side is bounded almost surely. Thus
$\{\nabla_{\boldsymbol\xi}W_R\}$ is bounded in
$L^2(Q;\R^3)$.

Let a subsequence of
$\nabla_{\boldsymbol\xi}W_R$ converge weakly in $L^2(Q;\R^3)$.
For every
$\boldsymbol\varphi\in C_c^\infty(Q;\R^3)$,
\begin{equation*}
\begin{aligned}
    \int_Q
    \nabla_{\boldsymbol\xi}W_R
    \cdot\boldsymbol\varphi
    \dd\boldsymbol\xi
    =
    -
    \int_Q
    W_R
    \nabla_{\boldsymbol\xi}\cdot\boldsymbol\varphi
    \dd\boldsymbol\xi
    \longrightarrow0
\end{aligned}
\end{equation*}
by \cref{eq:appendix_WR_zero}. Hence every weakly convergent
subsequence has limit $\bm0$, and therefore
\begin{equation}
    \nabla_{\boldsymbol\xi}
    \left[
    \frac{h(R\boldsymbol\xi,\omega)}{R}
    \right]
    \rightharpoonup
    \bm0
    \qquad
    \text{weakly in }L^2(Q;\R^3).
    \label{eq:appendix_weak_WR_field}
\end{equation}

Subtracting
\cref{eq:appendix_weak_WR_field} from
\cref{eq:appendix_weak_VR_field}, restricted to $Q$, gives
\begin{equation*}
    \nabla_{\boldsymbol\xi}
    \left[
    \frac{
    h_R(R\boldsymbol\xi,\omega)
    -
    h(R\boldsymbol\xi,\omega)
    }{R}
    \right]
    \rightharpoonup
    \nabla H_Q[\bp](\boldsymbol\xi)
\end{equation*}
weakly in $L^2(Q;\R^3)$, proving
\cref{eq:weak_finite_volume_field}.
\end{proof}

\section{Electrostatic identities}
\label{app:electrostatic_identities}

\subsection{Depolarization identity for the unit cube}
\label{app:depolarization_identity}

For $\bq\in\R^3$, define the linear map
\begin{equation*}
    \bD_Q\bq
    :=
    \int_Q
    \nabla H_Q[\bq](\boldsymbol\xi)
    \dd\boldsymbol\xi.
\end{equation*}
For the coordinate vector $\be_j$, the equation
\begin{equation*}
    -\Delta H_Q[\be_j]
    =
    -\nabla\cdot(\be_j\chi_Q)
\end{equation*}
gives
\begin{equation*}
    \widehat{H_Q[\be_j]}(\boldsymbol\kappa)
    =
    -\frac{\mathrm i\kappa_j}{|\boldsymbol\kappa|^2}
    \widehat{\chi_Q}(\boldsymbol\kappa).
\end{equation*}
Hence, by Plancherel's identity,
\begin{equation}
    (\bD_Q)_{ij}
    =
    \frac{1}{(2\pi)^3}
    \int_{\R^3}
    \frac{\kappa_i\kappa_j}{|\boldsymbol\kappa|^2}
    \left|
    \widehat{\chi_Q}(\boldsymbol\kappa)
    \right|^2
    \dd\boldsymbol\kappa.
    \label{eq:cube_depolarization_fourier}
\end{equation}
Reflection symmetry of $Q$ makes the off-diagonal entries zero, while
invariance under permutations of the coordinate axes makes the three
diagonal entries equal. Thus
\begin{equation*}
    \bD_Q=c\bI.
\end{equation*}
Taking the trace in \cref{eq:cube_depolarization_fourier} gives
\begin{equation*}
    3c
    =
    \frac{1}{(2\pi)^3}
    \int_{\R^3}
    \frac{\kappa_1^2+\kappa_2^2+\kappa_3^2}
    {|\boldsymbol\kappa|^2}
    \left|
    \widehat{\chi_Q}(\boldsymbol\kappa)
    \right|^2
    \dd\boldsymbol\kappa
    =
    \frac{1}{(2\pi)^3}
    \int_{\R^3}
    \left|
    \widehat{\chi_Q}(\boldsymbol\kappa)
    \right|^2
    \dd\boldsymbol\kappa
    =
    \int_{\R^3}
    |\chi_Q(\boldsymbol\xi)|^2
    \dd\boldsymbol\xi
    =
    |Q|
    =
    1.
\end{equation*}
Therefore, $\bD_Q=\bI/3$.

The weak equation for $H_Q[\bq]$ is
\begin{equation*}
    \int_{\R^3}
    \nabla H_Q[\bq]\cdot\nabla\psi
    \dd\boldsymbol\xi
    =
    \int_Q
    \bq\cdot\nabla\psi
    \dd\boldsymbol\xi
    \qquad
    \text{for every }\psi\in\dot H^1(\R^3).
\end{equation*}
Taking $\psi=H_Q[\bq]$ and using
$\bD_Q=\bI/3$ gives
\begin{equation}
    \int_{\R^3}
    |\nabla H_Q[\bq]|^2
    \dd\boldsymbol\xi
    =
    \int_Q
    \bq\cdot\nabla H_Q[\bq]
    \dd\boldsymbol\xi
    =
    \bq\cdot\bD_Q\bq
    =
    \frac13|\bq|^2.
\label{eq:cube_depolarization_energy}
\end{equation}


\subsection{Distributional form of the dipole kernel}
\label{app:dipole_distribution_identity}

Because the dipole kernel has an $|\bz|^{-3}$ singularity, it is not
locally integrable at the origin and must be interpreted as a distribution. This distributional treatment, including the Dirac contribution associated with the singularity, also appears in lattice-dipole calculations; see \cite[Sections~2--3]{JamesMuller1994}. We have, for $\bz\ne\bm0$,
\begin{equation*}
    \bK(\bz)
    =
    -\nabla^2G(\bz)
    =
    \frac{1}{4\pi|\bz|^3}
    \left(
    \bI
    -
    3\frac{\bz}{|\bz|}
    \otimes
    \frac{\bz}{|\bz|}
    \right).
\end{equation*}
We first note that $\bK$ has zero spherical mean. Consequently, the principal-value distribution associated with $\bK$ is well defined by
\begin{equation}
    \langle\pv\bK,\boldsymbol\varphi\rangle
    :=
    \lim_{\delta\downarrow0}
    \int_{|\bz|>\delta}
    \bK(\bz):\boldsymbol\varphi(\bz)
    \dd\bz,
    \qquad
    \boldsymbol\varphi\in
    C_c^\infty(\R^3;\R^{3\times3}).
    \label{eq:principal_value_dipole_kernel}
\end{equation}
To see the convergence near the origin, write
\begin{equation*}
    \boldsymbol\varphi(\bz)
    =
    \boldsymbol\varphi(\bm0)
    +
    \bigl[
    \boldsymbol\varphi(\bz)
    -
    \boldsymbol\varphi(\bm0)
    \bigr].
\end{equation*}
The contribution of $\boldsymbol\varphi(\bm0)$ vanishes on every
spherical annulus due to zero spherical mean of $\bK$, while
\begin{equation*}
    \left|
    \boldsymbol\varphi(\bz)
    -
    \boldsymbol\varphi(\bm0)
    \right|
    \leq
    C|\bz|
\end{equation*}
near the origin. Since $|\bK(\bz)|\leq C|\bz|^{-3}$, the remaining
integrand is bounded by $C|\bz|^{-2}$ and is locally integrable in
$\R^3$.

For $\bz\ne\bm0$,
\begin{equation*}
    -\nabla^2G(\bz)=\bK(\bz).
\end{equation*}
Hence the distribution
\begin{equation*}
    -\nabla^2G-\pv\bK
\end{equation*}
is supported at the origin. Both $-\nabla^2G$ and $\pv\bK$ are
homogeneous of degree $-3$. A distribution supported at the origin is a
finite linear combination of derivatives of $\delta_{\bm0}$, and
$\partial^\alpha\delta_{\bm0}$ is homogeneous of degree
$-3-|\alpha|$. Therefore,
\begin{equation}
    -\nabla^2G-\pv\bK
    =
    \bA\delta_{\bm0}
    \label{eq:dipole_distribution_point_term}
\end{equation}
for some constant matrix $\bA$. The left-hand side of
\cref{eq:dipole_distribution_point_term} is invariant under rotations. Therefore, $\bA=c\bI$ for some constant $c$. Taking the trace in \cref{eq:dipole_distribution_point_term} and using
\begin{equation*}
    -\Delta G=\delta_{\bm0},
    \qquad
    \operatorname{tr}\bK(\bz)=0
    \quad
    \text{for }\bz\ne\bm0,
\end{equation*}
gives $c=1/3$. The same $\bI/3$ Dirac correction is obtained for the spherically truncated dipole field in \cite[Proposition~3.1]{JamesMuller1994}. Hence,
\begin{equation}
    -\nabla^2G
    =
    \pv\bK
    +
    \frac13\bI\delta_{\bm0}
    \qquad
    \text{in }
    \mathcal D'(\R^3;\R^{3\times3}).
    \label{eq:dipole_distribution_identity}
\end{equation}

\section{Auxiliary results related to examples}

\subsection{Finite-volume fields for the periodic example}
\label{app:periodic_finite_volume_fields}

We verify \cref{eq:periodic_A3_local,eq:periodic_A3_nonlocal}.
Fix $\bx\in\overline\Omega$ and suppress $\bx$ in the notation. Define
the potential and field generated by one microscopic cell by
\begin{equation}
    \Psi(\by)
    :=
    \int_Q
    G(\by-\bs)
    \rho_{\mathrm{per}}(\bs)
    \dd\bs,
    \qquad
    \bF(\by)
    :=
    \nabla\Psi(\by).
    \label{eq:periodic_single_cell_field}
\end{equation}
The cell is neutral and has dipole moment
\begin{equation*}
    \bp
    =
    \int_Q
    \bs\,\rho_{\mathrm{per}}(\bs)
    \dd\bs.
\end{equation*}
Let
\begin{equation}
    \bD(\by)
    :=
    \bF(\by)-\nabla H_Q[\bp](\by).
    \label{eq:periodic_cell_field_difference}
\end{equation}
The two sources defining the fields in
\cref{eq:periodic_cell_field_difference} have the same total charge and
the same dipole moment. Taylor expansion of the Coulomb kernel therefore
gives
\begin{equation}
    |\bD(\by)|
    \leq
    \frac{C}{1+|\by|^4},
    \qquad
    |\by|\geq2\operatorname{diam}Q,
    \label{eq:periodic_cell_field_decay}
\end{equation}
where $C$ is independent of $\bx$. The assumed uniform support away from
$\partial Q$ and uniform boundedness of $\rho_{\mathrm{per}}$ also give
$\bF\in L^2_{\mathrm{loc}}(\R^3)$, while
$\nabla H_Q[\bp]\in L^2(\R^3)$. Hence
\begin{equation}
    \|\bD\|_{L^2(\bm\ell+Q)}
    \leq
    \frac{C}{1+|\bm\ell|^4},
    \qquad
    \bm\ell\in\Z^3.
    \label{eq:periodic_cell_field_l2_decay}
\end{equation}
In particular, the lattice sum of $\bD$ converges absolutely in
$L^2_{\mathrm{loc}}$.

The uniformly polarized cells tile space, so their internal surface
charges cancel. Moreover, if $\bD=\nabla U$, the zero charge and zero
dipole of the source of $U$ give
$U(\by)=O(|\by|^{-3})$. Hence
\begin{equation*}
    \int_{\R^3}\bD(\by)\dd\by
    =
    \lim_{r\to\infty}
    \int_{\partial B_r}U(\by)\bn\,\dd S
    =
    \bm0.
\end{equation*}
The absolutely convergent lattice sum of $\bD$ is therefore periodic,
curl free, and has zero mean over $Q$. Its divergence gives
$-\Delta h_{\mathrm{per}}=\rho_{\mathrm{per}}$, and consequently
\begin{equation}
    \nabla h_{\mathrm{per}}(\by)
    =
    \sum_{\bj\in\Z^3}
    \bD(\by-\bj).
    \label{eq:periodic_field_D_representation}
\end{equation}

For the scaled cell $Q_R=RQ$, define
\begin{equation}
    H_{Q_R}[\bp](\by)
    :=
    \int_{\partial Q_R}
    G(\by-\bz)
    \bp\cdot\bn_{Q_R}(\bz)
    \dd S_{\bz}.
    \label{eq:polarized_scaled_cell_potential}
\end{equation}
Equivalently,
\begin{equation*}
    -\Delta H_{Q_R}[\bp]
    =
    -\nabla\cdot(\bp\chi_{Q_R})
    \qquad
    \text{in }\mathcal D'(\R^3).
\end{equation*}
Since $Q_R=RQ$ and $G$ is homogeneous of degree $-1$,
\begin{equation}
    H_{Q_R}[\bp](R\boldsymbol\xi)
    =
    R H_Q[\bp](\boldsymbol\xi),
    \qquad
    \nabla H_{Q_R}[\bp](R\boldsymbol\xi)
    =
    \nabla H_Q[\bp](\boldsymbol\xi).
    \label{eq:polarized_scaled_cell_scaling}
\end{equation}

For the finite block $Q_R$, linearity and cancellation of the surface
charges on the internal faces of the uniformly polarized cells give
\begin{equation}
    \nabla h_R(\by)
    =
    \sum_{\bj\in I_R}
    \bF(\by-\bj)=
    \sum_{\bj\in I_R}
    \bD(\by-\bj)
    +
    \nabla H_{Q_R}[\bp](\by).
    \label{eq:periodic_finite_block_D_decomposition}
\end{equation}

We first consider $\boldsymbol\xi\in Q$. Combining
\cref{eq:periodic_field_D_representation,eq:periodic_finite_block_D_decomposition,eq:polarized_scaled_cell_scaling} gives
\begin{equation}
    \nabla_{\boldsymbol\xi}
    \left[
    \frac{
    h_R(R\boldsymbol\xi)
    -
    h_{\mathrm{per}}(R\boldsymbol\xi)
    }{R}
    \right]
    -
    \nabla H_Q[\bp](\boldsymbol\xi) =
    -
    \sum_{\bj\notin I_R}
    \bD(R\boldsymbol\xi-\bj).
    \label{eq:periodic_local_remainder}
\end{equation}
Let
\begin{equation*}
    a_{\bm\ell}
    :=
    \|\bD\|_{L^2(\bm\ell+Q)}.
\end{equation*}
By \cref{eq:periodic_cell_field_l2_decay},
\begin{equation*}
    a_{\bm\ell}
    \leq
    \frac{C}{1+|\bm\ell|^4},
    \qquad
    \sum_{\bm\ell\in\Z^3}a_{\bm\ell}<\infty.
\end{equation*}
For an observation cell $\bi+Q\subset Q_R$, let
$d(\bi)$ denote its lattice distance from the complement of $I_R$.
The triangle inequality gives
\begin{equation*}
    \left\|
    \sum_{\bj\notin I_R}
    \bD(\cdot-\bj)
    \right\|_{L^2(\bi+Q)}
    \leq
    \sum_{\bj\notin I_R}a_{\bi-\bj} \leq
    C
    \sum_{m\geq d(\bi)}
    \frac{m^2}{1+m^4} \leq
    \frac{C}{1+d(\bi)}.
\end{equation*}
There are $O(R^2)$ cells at each fixed depth from $\partial Q_R$.
Consequently,
\begin{equation}
    \frac{1}{R^3}
    \int_{Q_R}
    \left|
    \sum_{\bj\notin I_R}
    \bD(\by-\bj)
    \right|^2
    \dd\by
    \leq
    \frac{C}{R}
    \sum_{d=0}^{R}
    \frac{1}{(1+d)^2} \leq
    \frac{C}{R}.
    \label{eq:periodic_local_remainder_estimate}
\end{equation}
By the change of variables $\by=R\boldsymbol\xi$,
\cref{eq:periodic_local_remainder,eq:periodic_local_remainder_estimate}
prove \cref{eq:periodic_A3_local}.

For the second limit, let
$\bk\in\Z^3\setminus\{\bm0\}$ be fixed. From
\cref{eq:periodic_finite_block_D_decomposition,eq:polarized_scaled_cell_scaling},
\begin{equation}
    \nabla_{\boldsymbol\xi}
    \left[
    \frac{h_R(R\boldsymbol\xi)}{R}
    \right]
    -
    \nabla H_Q[\bp](\boldsymbol\xi) =
    \sum_{\bj\in I_R}
    \bD(R\boldsymbol\xi-\bj),
    \qquad
    \boldsymbol\xi\in\bk+Q.
    \label{eq:periodic_nonlocal_remainder}
\end{equation}
The interiors of $Q_R$ and $R(\bk+Q)$ are disjoint. For an observation
cell $\bi+Q$ at lattice distance $d$ from $Q_R$,
\cref{eq:periodic_cell_field_l2_decay} gives, as above,
\begin{equation*}
    \left\|
    \sum_{\bj\in I_R}
    \bD(\cdot-\bj)
    \right\|_{L^2(\bi+Q)}
    \leq
    \frac{C}{1+d}.
\end{equation*}
For fixed $\bk\ne\bm0$, there are at most $C_{\bk}R^2$ observation
cells at each distance $d$. Therefore,
\begin{equation}
    \frac{1}{R^3}
    \int_{R(\bk+Q)}
    \left|
    \sum_{\bj\in I_R}
    \bD(\by-\bj)
    \right|^2
    \dd\by
    \leq
    \frac{C_{\bk}}{R}
    \sum_{d=0}^{\infty}
    \frac{1}{(1+d)^2} \leq
    \frac{C_{\bk}}{R}.
    \label{eq:periodic_nonlocal_remainder_estimate}
\end{equation}
Together with \cref{eq:periodic_nonlocal_remainder}, this proves
\cref{eq:periodic_A3_nonlocal}.

\subsection{Fields for the random-displacement example}
\label{app:random_displacement_fields}

For fixed $\omega$, continuity of translations in $L^2(\R^3)$ and the
continuity of $\eta_\alpha$, $\ba_\alpha$, and $\bu_{\bk,\alpha}$ imply
that $\bar\rho$ is continuous in $\bx$ in $L^2(Q)$. The uniform profile
bound and confinement to $Q$ permit dominated convergence in $\omega$,
giving continuity in $L^2(Q\times\Xi)$. Together with cell neutrality,
this proves \textbf{A1}.

To verify \textbf{A2}--\textbf{A3}, write
\begin{equation}
    \rho_{\bk}(\bx,\bs,\omega)
    :=
    \bar\rho(\bx,\bs,T_{\bk}\omega)
    =
    m(\bx,\bs)
    +
    \widetilde\rho_{\bk}(\bx,\bs,\omega),
    \qquad
    m(\bx,\bs)
    :=
    \E[\bar\rho(\bx,\bs,\cdot)].
    \label{eq:random_mean_centered_decomposition}
\end{equation}
Exact cell neutrality gives
\begin{equation}
    \E[\widetilde\rho_{\bk}(\bx,\bs,\cdot)]=0,
    \qquad
    \int_Q
    \widetilde\rho_{\bk}(\bx,\bs,\omega)\dd\bs=0.
    \label{eq:centered_random_charge_properties}
\end{equation}
For fixed $\bx$, the fields
$\{\widetilde\rho_{\bk}\}_{\bk\in\Z^3}$ are independent.

Define the potential generated by the centered charge in cell $\bk$ by
\begin{equation}
    V_{\bk}(\bx,\by,\omega)
    :=
    \int_Q
    G(\by-\bk-\bs)
    \widetilde\rho_{\bk}(\bx,\bs,\omega)
    \dd\bs.
    \label{eq:centered_random_cell_potential}
\end{equation}
The uniform profile bound and local integrability of
$G$ and $\nabla G$ control the potential and field near the source cell. Cell neutrality gives their decay at infinity. Thus
\begin{equation}
    |V_{\bk}(\bx,\by,\omega)|
    \leq
    \frac{C}{1+|\by-\bk|^2},
    \qquad
    |\nabla V_{\bk}(\bx,\by,\omega)|
    \leq
    \frac{C}{1+|\by-\bk|^3},
    \label{eq:centered_random_cell_decay}
\end{equation}
for all $\by\in\R^3$, with $C$ independent of
$\bx$, $\bk$, and $\omega$. Consequently,
\begin{equation}
\begin{aligned}
    \sum_{\bk\in\Z^3}
    \E\!\left[
    \|V_{\bk}(\bx,\cdot,\cdot)\|_{L^2(Q)}^2
    \right]
    &\leq
    C
    \sum_{\bk\in\Z^3}
    \frac{1}{(1+|\bk|)^4}
    <
    \infty.
\end{aligned}
    \label{eq:random_potential_square_sum}
\end{equation}
The variables $V_{\bk}(\bx,\cdot,\cdot)$ are independent and centered in
$L^2(Q)$. Hence \cref{eq:random_potential_square_sum} implies that
\begin{equation*}
    \widetilde h(\bx,\cdot,\omega)
    :=
    \sum_{\bk\in\Z^3}V_{\bk}(\bx,\cdot,\omega)
\end{equation*}
converges in $L^2(Q\times\Xi)$ and almost surely in $L^2(Q)$. Applying the
same argument on every translated cell and taking a countable intersection
defines $\widetilde h\in L^2_{\mathrm{loc}}(\R^3)$ almost surely. The
convergence in $L^2(Q\times\Xi)$ is unconditional, so reindexing the series
under $T_{\bi}$ gives
\begin{equation*}
    \widetilde h(\bx,\by+\bi,\omega)
    =
    \widetilde h(\bx,\by,T_{\bi}\omega),
    \qquad \bi\in\Z^3.
\end{equation*}
Let $h_m$ be the zero-mean periodic potential of $m$, given by
\cref{eq:periodic_coulomb_convolution} with $\rho_{\mathrm{per}}=m$.
Then $h=h_m+\widetilde h$ is lattice stationary and has the integrability
required in \textbf{A2}. The uniform estimates above and continuity of
the charge profiles in $L^2(Q\times\Xi)$ give continuity of $h$ in
$L^2(Q\times\Xi)$ by dominated convergence. Cauchy--Schwarz then gives
continuity of the local observable in \textbf{A2}.

The same decomposition identifies the microscopic energy expectation.
Fix $\bx$ and suppress it in the following formulas. Since
$\bar\rho=m+\widetilde\rho_{\bm0}$ and
$h=h_m+\sum_{\bk}V_{\bk}$, we obtain
\begin{equation}
    \E\!\left[\int_Q\bar\rho\,\bar h\dd\bs\right]
    =
    \int_Q m h_m\dd\bs
    +
    \E\!\left[
    \int_Q\widetilde\rho_{\bm0}(\bs,\cdot)
    V_{\bm0}(\bs,\cdot)\dd\bs
    \right].
    \label{eq:random_cell_energy_expectation}
\end{equation}
Indeed, the mixed mean and fluctuation terms have zero expectation.
For $\bk\ne\bm0$, independence and centering give
$\E[\int_Q\widetilde\rho_{\bm0}V_{\bk}\dd\bs]=0$.
The $L^2(Q\times\Xi)$ convergence of the potential series justifies
passing to the limit in these pairings. The remaining term is the
free-space Coulomb contribution
\begin{equation*}
    \int_Q\widetilde\rho_{\bm0}V_{\bm0}\dd\bs
    =
    \int_Q\int_Q
    G(\bs-\boldsymbol t)
    \widetilde\rho_{\bm0}(\bs,\omega)
    \widetilde\rho_{\bm0}(\boldsymbol t,\omega)
    \dd\boldsymbol t\dd\bs.
\end{equation*}
For the law in \cref{sec:numerical_illustration}, its expectation is
the arithmetic mean over the four possible pairs
$(S_{\bm0,+},S_{\bm0,-})$.

It remains to verify \textbf{A3}. Let $\widetilde h_R$ and
$\widetilde h$ denote the finite- and infinite-volume potentials generated
by the centered fields. On $Q$, the local correction has the exact
representation
\begin{equation*}
    \nabla_{\boldsymbol\xi}
    \left[
    \frac{
    \widetilde h_R(\bx,R\boldsymbol\xi,\omega)
    -
    \widetilde h(\bx,R\boldsymbol\xi,\omega)
    }{R}
    \right]
    =
    -\sum_{\bj\notin I_R}
    \nabla V_{\bj}(\bx,R\boldsymbol\xi,\omega).
\end{equation*}
The summands are independent and centered as $L^2(Q;\R^3)$-valued random
variables, so their cross terms vanish after taking expectation.
Consequently, \cref{eq:centered_random_cell_decay} gives
\begin{equation}
    \E\!\left[
    \left\|
    \nabla_{\boldsymbol\xi}
    \left[
    \frac{
    \widetilde h_R(\bx,R\boldsymbol\xi,\cdot)
    -
    \widetilde h(\bx,R\boldsymbol\xi,\cdot)
    }{R}
    \right]
    \right\|_{L^2(Q)}^2
    \right]
    \leq
    \frac{C}{R}.
    \label{eq:random_local_field_second_moment}
\end{equation}
Indeed, for an observation cell at depth $d$ from $\partial Q_R$, the
sum of the squared fields generated by exterior cells is bounded by
\begin{equation*}
    C\sum_{m\geq d}\frac{m^2}{1+m^6}
    \leq
    \frac{C}{(1+d)^3}.
\end{equation*}
There are $O(R^2)$ observation cells at each fixed depth, which gives
\cref{eq:random_local_field_second_moment} after division by $R^3$.

On a translated region $\bk+Q$, the corresponding representation is
\begin{equation*}
    \nabla_{\boldsymbol\xi}
    \left[
    \frac{
    \widetilde h_R(\bx,R\boldsymbol\xi,\omega)
    }{R}
    \right]
    =
    \sum_{\bj\in I_R}
    \nabla V_{\bj}(\bx,R\boldsymbol\xi,\omega).
\end{equation*}
The same argument gives, for every fixed
$\bk\in\Z^3\setminus\{\bm0\}$,
\begin{equation}
    \E\!\left[
    \left\|
    \nabla_{\boldsymbol\xi}
    \left[
    \frac{
    \widetilde h_R(\bx,R\boldsymbol\xi,\cdot)
    }{R}
    \right]
    \right\|_{L^2(\bk+Q)}^2
    \right]
    \leq
    \frac{C_{\bk}}{R}.
    \label{eq:random_nonlocal_field_second_moment}
\end{equation}

To obtain almost-sure convergence, set
\begin{equation*}
\begin{aligned}
    X_{\bj}^{R,\mathrm{loc}}(\boldsymbol\xi)
    &:=-\nabla V_{\bj}(\bx,R\boldsymbol\xi,\cdot),
    &&\bj\notin I_R,
    \quad \boldsymbol\xi\in Q,\\
    X_{\bj}^{R,\bk}(\boldsymbol\xi)
    &:=\nabla V_{\bj}(\bx,R\boldsymbol\xi,\cdot),
    &&\bj\in I_R,
    \quad \boldsymbol\xi\in\bk+Q.
\end{aligned}
\end{equation*}
Each family consists of independent centered $L^2$-valued random
variables, and its sum is the corresponding field above. For
either family, expanding the fourth moment and using independence gives
\begin{equation*}
    \E\!\left[
    \left\|
    \sum_{\bj}X_{\bj}
    \right\|_{L^2}^4
    \right]
    \leq
    C
    \left[
    \left(
    \sum_{\bj}
    \E[\|X_{\bj}\|_{L^2}^2]
    \right)^2
    +
    \sum_{\bj}
    \E[\|X_{\bj}\|_{L^2}^4]
    \right].
\end{equation*}
The shell estimates used above also give
\begin{equation*}
    \sum_{\bj}
    \sup_{\omega\in\Xi}
    \|X_{\bj}\|_{L^2}^2
    \leq
    \frac{C}{R},
\end{equation*}
where $C$ is replaced by $C_{\bk}$ for the translated family. Consequently,
\begin{equation*}
\begin{aligned}
    \left(
    \sum_{\bj}
    \E[\|X_{\bj}\|_{L^2}^2]
    \right)^2
    &\leq \frac{C}{R^2},\\
    \sum_{\bj}
    \E[\|X_{\bj}\|_{L^2}^4]
    &\leq
    \left(
    \sum_{\bj}
    \sup_{\omega\in\Xi}
    \|X_{\bj}\|_{L^2}^2
    \right)^2
    \leq \frac{C}{R^2}.
\end{aligned}
\end{equation*}
Substituting these bounds into the fourth-moment estimate gives
\begin{equation}
\begin{aligned}
    \E\!\left[
    \left\|
    \nabla_{\boldsymbol\xi}
    \left[
    \frac{
    \widetilde h_R(\bx,R\boldsymbol\xi,\cdot)
    -
    \widetilde h(\bx,R\boldsymbol\xi,\cdot)
    }{R}
    \right]
    \right\|_{L^2(Q)}^4
    \right]
    &\leq
    \frac{C}{R^2},
    \\
    \E\!\left[
    \left\|
    \nabla_{\boldsymbol\xi}
    \left[
    \frac{
    \widetilde h_R(\bx,R\boldsymbol\xi,\cdot)
    }{R}
    \right]
    \right\|_{L^2(\bk+Q)}^4
    \right]
    &\leq
    \frac{C_{\bk}}{R^2}.
\end{aligned}
    \label{eq:random_field_fourth_moments}
\end{equation}
Fix $\bx\in\overline\Omega$ and
$\bk\in\Z^3\setminus\{\bm0\}$, and denote the two norms in
\cref{eq:random_field_fourth_moments} by
$Z_R^{\mathrm{loc}}$ and $Z_R^{\bk}$, respectively. For every
$\delta>0$, Markov's inequality and the union bound give
\begin{equation*}
\begin{aligned}
    \Pp\!\left(
    Z_R^{\mathrm{loc}}>\delta
    \ \text{or}\
    Z_R^{\bk}>\delta
    \right)
    &\leq
    \delta^{-4}
    \left(
    \E[(Z_R^{\mathrm{loc}})^4]
    +
    \E[(Z_R^{\bk})^4]
    \right)
    \leq
    \frac{C_{\bk}}{\delta^4R^2}.
\end{aligned}
\end{equation*}
Consequently, for every $m\in\mathbb N$,
\begin{equation*}
    \sum_{n=0}^{\infty}
    \Pp\!\left(
    Z_{2n+1}^{\mathrm{loc}}>\frac1m
    \ \text{or}\
    Z_{2n+1}^{\bk}>\frac1m
    \right)
    <
    \infty.
\end{equation*}
The Borel--Cantelli lemma and a countable intersection over
$m\in\mathbb N$ therefore give, along $R=2n+1$,
\begin{equation}
    \nabla_{\boldsymbol\xi}
    \left[
    \frac{
    \widetilde h_R(\bx,R\boldsymbol\xi,\omega)
    -
    \widetilde h(\bx,R\boldsymbol\xi,\omega)
    }{R}
    \right]
    \longrightarrow
    \bm0
    \quad\text{strongly in }L^2(Q;\R^3)
    \label{eq:random_centered_A3_local}
\end{equation}
and
\begin{equation}
    \nabla_{\boldsymbol\xi}
    \left[
    \frac{
    \widetilde h_R(\bx,R\boldsymbol\xi,\omega)
    }{R}
    \right]
    \longrightarrow
    \bm0
    \quad\text{strongly in }L^2(\bk+Q;\R^3)
    \label{eq:random_centered_A3_nonlocal}
\end{equation}
almost surely. The mean profile $m(\bx,\cdot)$ satisfies the periodic
hypotheses of \cref{sec:periodic_example}. Combining
\cref{eq:random_centered_A3_local,eq:random_centered_A3_nonlocal}
with
\cref{eq:periodic_A3_local,eq:periodic_A3_nonlocal}
proves \textbf{A3}.

\end{document}